\documentclass[a4paper,10pt,reqno]{amsart}
          \usepackage{amssymb}
	  \usepackage{amsmath}
          \usepackage{amsfonts}
          \usepackage[english]{babel}
          \usepackage[utf8]{inputenc}

\usepackage[margin=2.5cm]{geometry}

 \usepackage[unicode,colorlinks,plainpages=false,hyperindex=true,bookmarksnumbered=true,bookmarksopen=false,pdfpagelabels]{hyperref}
 \hypersetup{urlcolor=cyan,linkcolor=blue,citecolor=red,colorlinks=true}
 \usepackage{color}
\newtheorem{thm}{Theorem}[section]

\newtheorem{lem}[thm]{Lemma}
\newtheorem{prop}[thm]{Proposition}

\newtheorem{rem}[thm]{Remark}
\newtheorem{defn}[thm]{Definition}

\numberwithin{equation}{section}

\newcommand{\be}{\begin{equation}}
\newcommand{\ee}{\end{equation}}
\newcommand{\bea}{\begin{eqnarray}}
\newcommand{\eea}{\end{eqnarray}}
\newcommand{\ba}{\begin{aligned}}
\newcommand{\ea}{\end{aligned}}

\begin{document}

\title{Poisson reductions of master integrable systems on doubles of compact Lie groups}

\maketitle

\begin{center}

L. Feh\'er${}^{a,b}$

\medskip
${}^a$Department of Theoretical Physics, University of Szeged\\
Tisza Lajos krt 84-86, H-6720 Szeged, Hungary\\
e-mail: lfeher@physx.u-szeged.hu

\medskip
${}^b$Institute for Particle and Nuclear Physics\\
Wigner Research Centre for Physics\\
 H-1525 Budapest, P.O.B.~49, Hungary

\end{center}

\begin{abstract}
We consider three `classical doubles' of any semisimple, connected and simply connected  compact Lie group $G$:
the cotangent bundle, the Heisenberg
double and the internally fused quasi-Poisson double. On each double we identify a pair of `master integrable systems’
and investigate their Poisson reductions. In the simplest cotangent bundle case,
the reduction is defined by taking quotient
 by the cotangent lift of the conjugation action of $G$ on itself, and this  naturally generalizes
  to the other two doubles.  In each case, we derive explicit formulas for the reduced Poisson structure and equations
  of motion and find that they are associated with well known classical dynamical $r$-matrices.
   Our principal result is that we provide a unified treatment of a large family
  of reduced systems, which contains new models as well as examples of spin Sutherland and Ruijsenaars--Schneider models
  that were studied previously. We argue that on generic symplectic leaves of the Poisson quotients
   the reduced systems are integrable in the degenerate sense,
  although further work is required to prove this rigorously.
\end{abstract}

 \setcounter{tocdepth}{2}

 \tableofcontents

 \vspace{1cm}

%%%%%%%%%%%%%%%%%%%%%%%%%%%%%%%%%%%%%%%%%%%%%
%%%     SOME MACROS USED IN THE TEXT      %%%
%%%%%%%%%%%%%%%%%%%%%%%%%%%%%%%%%%%%%%%%%%%%%
\def\bI{\mathbb{I}}                         %
\def\span{\mathrm{span}}                    %
\def\1{{\boldsymbol 1}}                     %
\def\cD{{\mathcal D}}                       %
\def\cH{{\mathcal H}}                       %
\def\tr{\mathrm{tr}}                        %
\def\ri{{\rm i}}                            %
\def\bC{\mathbb{C}}                         %
\def\bR{\mathbb{R}}                         %
\def\bZ{\mathbb{Z}}                         %
\def\cF{{\mathcal F}}                       %
\def\reg{\mathrm{reg}}                      %
\def\id{{\mathrm{id}}}                      %
\def\dt {\left.\frac{d}{dt}\right|_{t=0}}   %
\def\fM{\mathfrak{M}}                       %
\def\cG{{\mathcal G}}                       %
\def\cR{{\mathcal R}}                       %
\def\cB{\mathcal{B}}                        %
\def\Dress{{\mathrm{Dress}}}                %
\def\dress{{\mathrm{dress}}}                %
\def\red{{\mathrm{red}}}                    %
\def\rank{{\mathrm{rank}}}                  %
\def\fP{\mathfrak{P}}                       %
\def\ad{\mathrm{ad}}                        %
\def\Ad{\mathrm{Ad}}                        %
\def\cA{\mathcal{A}}                        %
\def\cM{\mathcal{M}}                        %
\def\fN{\mathfrak{N}}                       %
\def\End{\mathrm{End}}                      %
\def\r{{\mathrm r}}                         %
\def\fR{\mathfrak{R}}                       %
\def\cC{\mathcal{C}}                        %
\def\cO{\mathcal{O}}                        %
\def\cW{\mathcal{W}}                        %
\def\fS{\mathfrak{D}}                       %
\def\fF{\mathfrak{F}}                       %
\def\cL{\mathcal{L}}                        %
\def\bfR{{\mathbf R}}                       %
\def\sl2z{\mathrm{SL}(2,\bZ)}               %
\def\cMpreg0{{{\cM}'}_0^\reg}               %
\def\fMpreg0{{{\fM}'}_0^\reg}               %
\def\fSpreg0{{{\fS}'}_0^\reg}               %
\def\CMpreg{{{\cM}'}^\reg}                  %
\def\FMpreg{{{\fM}'}^\reg}                  %
\def\FSpreg{{{\fS}'}^\reg}                  %
\def\cS{{\mathcal S}}                       %
%%%%%%%%%%%%%%%%%%%%%%%%%%%%%%%%%%%%%%%%%%%%%
%%%%%%%%%%%%%%%%%%%%%%%%%%%%%%%%%%%%%%%%%%%%%

\newpage

\section{Introduction}
\label{S:Sec1}

The variants of the method of Hamiltonian reduction \cite{A,OR,Rud} play a pivotal role
in deriving and analyzing integrable Hamiltonian systems.
The starting point in the applications is always a manifestly integrable
system on a higher dimensional phase space that possesses  a large symmetry group,
which is used for setting up its reduction.
As examples, it is sufficient to mention that key properties of the ubiquitous Calogero--Moser--Sutherland models
\cite{Cal,M,S}
and their relativistic \cite{RS} and spin generalizations \cite{GH,KZ,LX} became transparent from investigations
based on this method \cite{AO,CF,FK1,FGNR,KKS,Res1}.
For reviews of the subject, see  \cite{A,N,OP,PolR}.
Building on our experience gained from previous studies \cite{FFM,F1,F2,FA,FK2,FP} here we wish to explore a general set of
reductions of important families of unreduced `master systems'.

Let $G$ be a compact, connected and simply connected Lie group whose Lie algebra $\cG$ is simple.
In this paper we study Poisson reductions of three  phase spaces
associated with $G$. The first is the cotangent bundle
\be
\cM:= T^* G \simeq G \times \cG,
\ee
presented by means of right-trivialization and the identification $\cG^* \simeq \cG$.
Its Poisson--Lie generalization is the Heisenberg double \cite{STS}
\be
\fM:= G \times B,
\ee
which is obtained by combining the standard multiplicative Poisson structures on $G$
and its dual Poisson--Lie group $B$ into a symplectic structure.
This is a natural generalization since $T^*G$ is the Heisenberg double for $G$
equipped with the zero Poisson structure.
The third unreduced phase space is the so-called internally fused quasi-Poisson double \cite{AKSM}, denoted
\be
\fS:= G \times G,
\ee
that is closely related to the moduli space of flat $G$-connections on the punctured torus.
Each of these spaces carries a pair of degenerate integrable systems, and
reductions of those to integrable many-body models and their spin extensions
have already received considerable attention (see, e.g., \cite{F1,FP,Res1,Res2} and references therein).
The goal of this paper is to describe a very general reduction of these `master integrable systems' in all three cases.
We shall apply the same technique in our study of the distinct cases, and shall highlight the similarities between the resulting
reduced systems.
The principal case of our interest is the Heisenberg double $\fM$.
We include the cotangent bundle in our treatment mainly in order to motive the generalizations,
although  new results will be obtained also in this familiar case.
The unified treatment that we present has not yet been developed in the literature,
and could be useful for further detailed explorations of the reduced systems descending
from the three doubles.

The doubles of $G$ are $G$-manifolds, where $\cM$ carries the cotangent lift of the
conjugation action of $G$ on itself, $G$ acts on $\fS$ by diagonal conjugations,
and there is a similar action on $\fM$ built from the conjugation action and
the dressing action of $G$ on $B$.
The Poisson brackets on $\cM$ and $\fM$ and the quasi-Poisson bracket on $\fS$ share
the property that the $G$-invariant smooth functions form a closed Poisson algebra.
By Poisson reduction, we mean the restriction to this Poisson algebra of invariant functions,
which is to be thought of as a Poisson structure on the corresponding quotient space
defined by the $G$-action.
The first principal goal of our work is to derive an effective description of these `reduced
Poisson algebras'.

Denote $C^\infty(G)^G$, $C^\infty(\cG)^G$ and $C^\infty(B)^G$ the respective rings of
invariant real functions. The functional dimension of these rings of functions equals the rank $\ell$ of $\cG$.
 All three doubles are Cartesian products as manifolds,
and we let $\pi_1$ and $\pi_2$  denote the projections onto the first and second factors of those Cartesian products.
Then, for each of the three doubles,
 $\pi_1^*(C^\infty(G)^G)$ provides an Abelian Poisson subalgebra
of the Poisson algebra of the $G$-invariant functions.
We call the elements of $\pi_1^*(C^\infty(G)^G)$ \emph{pullback invariants}.
Using $\pi_2$ in the analogous manner, one also obtains Abelian Poisson algebras of pullback invariants.
The Poisson and quasi-Poisson structures allow one to associate a (Hamiltonian or quasi-Hamiltonian)
vector field to every function, defining an evolution equation.
We shall explain that the evolution equation obtained from any pullback
invariant gives rise to a degenerate integrable system \cite{MF,Nekh,Res2}, which means\footnote{Here, we implicitly extended Definition \ref{defn:degint} below
to the non-symplectic case of the quasi-Poisson double.}
that it admits a ring of
constants of motion whose functional dimension is equal to $2 \dim(G) - \ell$, where $2\dim(G)$
is the dimension of the phase space.
We shall also explicitly describe the integral curves of the pullback invariants and their
constants of motion in each case.  This yields generalizations of well-known results concerning
$T^*G$.
Our second principal goal is to characterize the reductions of the degenerate
integrable systems induced on the master phase spaces by the pullback invariants.

For clarity, recall  that the functional dimension of a ring of smooth functions $\fF$ on a manifold $X$ is $k$ if
there exists an open dense submanifold $\tilde X \subseteq X$ such that the exterior derivatives of the
elements of $\fF$ span a $k$-dimensional subspace of $T^*_xX$ for every $x\in \tilde X$.
The fact that the rings of invariants of our concern have functional dimension $\ell = \rank(\cG)$ follows
from basic Lie theoretic results,
and the pullback invariants obviously have the same functional
dimension as the original invariants.  Below, the functional dimension of a Poisson algebra is understood
to mean the functional dimension of the underlying ring of functions.

The quotient spaces of the master phase spaces are not smooth manifolds, but stratified
Poisson spaces \cite{OR,SL,Sn},  which still can be decomposed into disjoint unions of smooth symplectic leaves.
However, this is  quite a complicated structure, and we will be content with describing the
Poisson algebras of the invariants, and the reductions of the evolution equations generated by
the pullback invariants,
in terms of convenient partial gauge fixings.
To explain what this means, we next outline the case of the cotangent bundle.
We then briefly summarize how the picture generalizes to the other cases.

\subsection*{The motivating example of $T^*G$ and its generalizations}
Let us fix a maximal torus $G_0 < G$ and let $\cG_0 < \cG$ be its Lie algebra.
The group $G$ acts on itself by conjugations and on $\cG$ by the adjoint action.
We denote by $G^\reg$ and $\cG^\reg$ the dense open subsets formed by the elements whose
isotropy subgroups are maximal tori in $G$, and let $G_0^\reg$ and $\cG_0^\reg$ be their intersections with $G_0$ and $\cG_0$, respectively.
Then the $G$-orbits through the submanifolds
\be
\cM_0^\reg :=\{ (Q,J) \in \cM \mid Q\in G_0^\reg\}
\quad\hbox{and}\quad
{\cM'}_0^\reg:=\{ (g, \lambda)\in \cM \mid \lambda \in \cG_0^\reg\}
\label{M012}\ee
fill dense open subsets of $\cM$, denoted $\cM^\reg$ and
$\CMpreg$.
The restriction of functions leads to isomorphisms
\be
C^\infty(\cM^\reg)^G \Longleftrightarrow  C^\infty(\cM_0^\reg)^\fN
\quad\hbox{and}\quad
C^\infty(\CMpreg)^G \Longleftrightarrow  C^\infty(\cMpreg0)^\fN,
\label{isom12}\ee
where $\fN < G$ denotes the normalizer of $G_0$ inside $G$.
Speaking colloquially, we say that
$\cM_0^\reg$ and $\cMpreg0$ provide partial gauge fixings for the $G$-action on
the dense open submanifolds  $\cM^\reg\subset \cM$ and $ \CMpreg\subset \cM$, and $\fN$ is
the corresponding residual gauge group.
The key point of our work is that we use the isomorphisms \eqref{isom12} of the respective rings of functions
to transfer the Poisson bracket of the $G$-invariant functions to the rings $C^\infty(\cM_0^\reg)^\fN$ and
$C^\infty(\cMpreg0)^\fN$. By definition, this gives the `reduced Poisson algebras'
\be
\left( C^\infty(\cM_0^\reg)^\fN, \{-,- \}_\red\right)
\quad \hbox{and}\quad
\left(C^\infty(\cMpreg0)^\fN , \{- ,- \}'_\red\right).
\label{redPB12}\ee
Since any smooth, even continuous, function can be recovered from its restriction to a dense open
subset, these Poisson algebras furnish two convenient descriptions
of the Poisson brackets of the elements of $C^\infty(\cM)^G$.
Their explicit formulas are given by Theorems \ref{thm:spinSuth} and \ref{thm:spinRS}, the former is well known,
while the latter seems to have escaped attention previously.

Here, a few clarifying remarks are in order. First, it should be noted that the reduced
Poisson algebras \eqref{redPB12} are larger than $(C^\infty(\cM)^G, \{- ,- \})$, since not every
smooth invariant function on a dense open subset extends to a smooth function on the full of $\cM$.
Second, these Poisson algebras can be extended to $G_0$-invariant smooth functions, for
$\cM_0^\reg/G_0$ is a covering space of $\cM_0^\reg/\fN$ with the fiber given by the Weyl group $\fN/G_0$;
and similar for  $\cMpreg0$.
Without resorting  to further ad hoc gauge fixings, it appears
difficult to gain more effective descriptions of the Poisson algebras of the invariant functions.

Now what can we say about the reductions of the two integrable systems on $\cM$?
Take an arbitrary function $\varphi \in C^\infty(\cG)^G$ and consider the restriction of
the pullback invariant $\pi_2^*(\varphi)$ to $\cM_0^\reg$. This `reduced Hamiltonian'
 defines a derivation
of the elements of $C^\infty(\cM_0^\reg)^\fN$ through the Poisson bracket $\{- ,- \}_\red$.
This derivation can be presented as a vector field on $\cM_0^\reg$, which
then gives rise to a `reduced evolution equation' on $\cM_0^\reg$.
We find (Proposition \ref{prop:redeq1}) that the resulting evolution equation takes the following form:
\be
\dot{Q} = (d\varphi(J))_0 Q, \qquad \dot{J} = [\cR(Q) d\varphi(J), J].
\label{redeq1I}\ee
Here, $d\varphi$ denotes the $\cG$-valued  gradient of $\varphi$, the subscript zero refers to
the orthogonal decomposition $\cG = \cG_0 + \cG_\perp$, and $\cR(Q)\in \End(\cG)$ is a well-known
trigonometric solution of the modified classical dynamical Yang--Baxter equation \cite{EV}.
It vanishes on $\cG_0$ and, writing $Q= \exp(\ri q)$ with $q\in \ri \cG_0^\reg$, it is given by
$\cR(Q) = \frac{1}{2}\coth(\frac{\ri}{2} \ad_q)$ on $\cG_\perp$.
(Here, $\ri \cG_0$ is a subset of the complexification of $\cG$.)
Of course, the so-obtained vector fields and evolution equations are unique only up to the addition
of arbitrary vector fields that are tangent to the $G_0$-orbits in $\cM_0^\reg$,
which generate infinitesimal residual gauge transformations.
This ambiguity drops
 out under the eventual projection to the reduced phase space $\cM/G$.
 Thus, our slight abuse of the term \emph{reduced} is harmless.

Similarly, the pullback invariants $\pi_1^*(h)$ associated with the functions $h\in C^\infty(G)^G$ lead to
 interesting reduced evolution equations on  ${\cM'}_0^\reg$. We find (Proposition \ref{prop:redeq2}) that they take the following form:
 \be
\dot{g} = [g, r(\lambda) \nabla h(g)], \qquad \dot{\lambda} = -(\nabla h(g))_0.
\label{redeq2I}\ee
Using the Killing form $\langle -,- \rangle_\cG$ of $\cG$, $\nabla h(g)\in \cG$ is defined by the relation
 $\langle X, \nabla h(g)\rangle_\cG = \dt h(e^{tX} g)$ for all $X\in \cG$,
and $r(\lambda) \in \End(\cG)$ is
the rational dynamical $r$-matrix that vanishes on $\cG_0$
 and operates on
 $\cG_\perp$ as $(\ad_\lambda)^{-1}$.
 These evolution equations matter up to residual gauge transformations
 like in the case of \eqref{redeq1I}.

 By parametrizing $J$ in \eqref{redeq1I} according to
\be
J= -\ri p - \cR(Q) \xi -\frac{1}{2} \xi  \quad\hbox{with}\quad p = \ri \cG_0,\,\, \xi \in \cG_\perp,
\label{Jpar}\ee
and taking $\varphi(J) = - \frac{1}{2} \langle J, J\rangle_\cG$, the system \eqref{redeq1I} can be recognized as
a spin Sutherland system, for which the components of $q$ and $p$ form canonically conjugate pairs
and $\xi$ is a so-called collective spin variable \cite{FP,LX}. (See also
equation \eqref{spinSuthHam}.)
For $G=\mathrm{SU}(n)$, restriction to a small symplectic leaf in the reduced phase space gives
the trigonometric (spinless) Sutherland system \cite{KKS}. On the same symplectic leaf, but using a different
parametrization and the Hamiltonian $h(g) = \Re \tr(g)$, the system \eqref{redeq2I} yields a specific
real form of the rational Ruijsenaars--Schneider system, which enjoys a duality relation with the trigonometric Sutherland system
\cite{FA,FGNR}.

 The above sketched results about reductions of the cotangent bundle are known to experts,
 especially the reduced system described in terms of  $\cM_0^\reg$.
 In this paper we take the lead from this example and characterize the reductions
 of the Heisenberg double $\fM$ and the quasi-Poisson double $\fS$ in a similar manner.
 To highlight a key feature of these generalizations, note that the first model $\cM_0^\reg$ was obtained by `diagonalizing' the first one out
of the pairs of elements forming $\cM$, and the second model $\cMpreg0$ was obtained
by diagonalizing the second constituent of those pairs.
The pullback invariants built by using  $\pi_2^*$ then led to interesting reduced evolution equations on $\cM_0^\red$,
and those built on $\pi_1^*$ led to interesting evolution equations on $\cMpreg0$.
The situation turns out fully analogous for the reductions of the other two doubles.
 In particular, we shall derive two presentations of the Poisson algebras of the $G$-invariant functions,
 and describe the form of the interesting reduced evolution equations induced by the two rings of pullback
 invariants.  Concerning the Heisenberg double, these results are summarized by Theorem \ref{thm:defSuth} together with
 Proposition \ref{prop:REDeq1+}, and Theorem \ref{thm:RED2} with Proposition \ref{prop:REDeq2},
 which are tied with two partial gauge fixings akin to what is displayed in \eqref{M012} for $T^*G$.
 The analogous results pertaining to the quasi-Poisson double are formulated in Theorem \ref{thm:redqPB} and Proposition \ref{prop:redqeq}.
 \emph{These theorems and propositions constitute the main new results of the present paper.}

Motivated by the case of $T^*G$ \cite{Res1} and the results of \cite{FA,FK1,FK2,FGNR}, we say that the two kinds of reduced systems
that arise from the same double are in duality with each other. In the case of the quasi-Poisson double,
duality  actually becomes self-duality. The meaning of these dualities will be elaborated  in the text.

\subsection*{Degenerate integrability and reduction}
First of all, let us specify the precise notion of degenerate integrability used in this paper.

\begin{defn}\label{defn:degint} By definition \cite{Nekh},
a \emph{degenerate integrable system}
on a symplectic manifold of dimension $N$ consists of an Abelian Poisson subalgebra of
the Poisson algebra of smooth functions such that its functional dimension, $\delta$, is smaller than $N/2$,
and the functional dimension of its centralizer is $(N-\delta)$.
To put it more plainly, the system  is built on $1\leq \delta< N/2$ functionally independent, mutually Poisson commuting Hamiltonians
that admit $(N-\delta)$ functionally independent joint constants of motion. An additional requirement is
that the commuting Hamiltonians should possess complete flows.
\end{defn}

Degenerate integrability is a stronger property than Liouville integrability, which corresponds to the limiting case $\delta=N/2$.
For the structure of the systems having this property, see \cite{MF,Nekh,Res2,Rud}.
Further variants of the notion of integrability, as well as their extension to Poisson manifolds, and even
to Abelian Lie algebras of non-Hamiltonian vector fields,
are also discussed in the literature \cite{J,LMV,Zung}.

The restrictions of our reduced systems are expected to give degenerate integrable systems on generic
symplectic leaves of the quotient space of the double in each case.   Reshetikhin has argued \cite{Res1}  that this is the case
for the complex holomorphic analog of the
cotangent bundle $T^*G$, and his arguments can be adapted to the compact real
form. His joint paper with Arthamonov \cite{AR} leads to the same conclusion regarding the quasi-Poisson double.
It may well be that integrability holds on all symplectic leaves (with only Liouville integrability on exceptional leaves),
but  we cannot prove this at present. Nevertheless,  we deem it worthwhile to outline two mechanisms that point toward the heuristic statement that
`degenerate integrability is generically inherited by the reduced systems engendered by Poisson reduction'.

  Let $V$ be a $G$-invariant vector field on a $G$-manifold $X$.
  Equivalently,  if
$x(t)$ is an integral curve of $V$,  then $A_\eta(x(t))$
is also an integral curve for each $\eta \in G$,
where $A_\eta$ denotes the diffeomorphism of $X$ associated with $\eta\in G$.
Suppose now that $G$ is compact and denote by $d_G$ the probability Haar measure on $G$.
For any real function $\cF \in C^\infty(X)$ define the function $\cF^G\in C^\infty(X)^G$ by averaging
the functions $A_\eta^* \cF$ over $G$,
\be
\cF^G(x):= \int_G \cF(A_\eta(x)) d_G(\eta), \qquad \forall x\in X.
\label{aver}\ee
Clearly, if $\cF$ is a constant of motion for the vector field $V$, then
$\cF^G$ is a $G$-invariant constant of motion for $V$.
In \cite{Zung} averaging was used for arguing that, generically,
degenerate integrability  survives Poisson reduction.
In this work it was assumed that the $G$-action is generated by an equivariant
moment map into $\cG^*$.
However, the fine structure of the quotient space of $X$ was treated only
rather casually.  See also the review \cite{J}.
The averaging of the unreduced constants of motion is applicable
in all cases that we study.
The Hamiltonian vector fields of the pullback invariants
are invariant under the `conjugation action' of $G$ on the unreduced phase space, except for the pullback invariants from
$\pi_1^*(C^\infty(G)^G) \subset C^\infty(\fM)$.
In the latter case, $G$-invariance of the Hamiltonian vector fields holds with respect to an action that has the same orbits
as the conjugation action;
this is explained in Appendix \ref{sec:C}.

Now we formulate a second mechanism whereby integrability can descend to reduced systems.
We extracted this mechanism from the work of Reshetikhin \cite{Res1}.  It will turn out to be applicable
to all of our examples of interest.
We begin by listing a number of strong assumptions.
First, consider two $G$-manifolds $X$ and $Y$ for which
both quotient spaces $X/G$ and $Y/G$ are manifolds such that
$\pi_X: X\to X/G$ and $\pi_Y: Y \to Y/G$ are smooth submersions.
Second, suppose that $\Psi: X \to Y$ is a smooth, $G$-equivariant, surjective map.
Then, $ \Psi$ gives rise to a well-defined smooth, surjective map
$\Psi_\red: X/G \to Y/G$, for which
\be
\pi_Y \circ \Psi = \Psi_\red \circ \pi_X.
\label{Psired}\ee
Third, suppose that we have a  vector field $V$ on $X$ that is projectable to
a vector field $V_\red$ on $X/G$.
Coming to the crux, if we now assume that $\Psi$ is constant along the integral
curves of $V$, then we obtain that $\Psi_\red$ is constant along the integral
curves of $V_\red$. Indeed, this holds since the integral curves of $V_\red$ result by
applying $\pi_X$ to the integral curves of $V$.
In such a situation, $\Psi^*(C^\infty(Y))$ gives constants of motion for $V$ and
$\Psi_\red^*(C^\infty(Y/G))$ gives
constants of motion for $V_\red$.
In particular, the functional dimension of the ring of constants of motion for the
projected vector field $V_\red$  is at most $\dim(G)$ less than the dimension of $Y$.
Under favourable circumstances, this mechanism can be used to show the degenerate
integrability of the reduced system on $X/G$ that  descends from the commuting (Hamiltonian) vector fields
 of a degenerate integrable system on $X$.
 The unreduced commuting Hamiltonians must be $G$-invariant, and must remain independent after reduction.
 To put this mechanism into practice, one may have to restrict oneself to dense open submanifolds and
 to generic symplectic leaves of the quotient Poisson structure.
 This will become clear in the examples.

\subsection*{Layout and notations}
The organization of the rest of the paper is shown by the table of contents.
Sections \ref{S:Sec2}, \ref{S:Sec3} and \ref{S:Sec4} are devoted to the three doubles, starting from the cotangent bundle.
In each case, we first describe the unreduced phase space and its degenerate integrable systems,
and then turn to their reductions. We have already delineated the theorems and proposition that
contain our main new results.
These results and open problems are briefly discussed in Section \ref{S:Sec5}.
Three appendices are also included, which contain auxiliary material.
In particular, Appendix \ref{sec:A} summarizes some Lie theoretic background that the reader
may wish to look at before reading Section \ref{S:Sec3}.

Throughout the paper, our notations  `pretend' that we are dealing with matrix Lie groups. For example,
$\eta J \eta^{-1}$ in equation \eqref{Aeta} denotes the adjoint action of $\eta \in G$ on $J\in \cG$.
As another example, $Xg$ in \eqref{cGg} stands for the value at $g\in G$
of the right-invariant vector field on $G$
associated with  the element $X$ from the Lie algebra $\cG$ of  $G$.
Such matrix notations simplify many formulas considerably, and can be easily converted into more abstract notation if desired.
Then one can verify that our results are valid for abstract Lie groups as well.
Alternatively, one may employ faithful matrix representations of the underlying Lie groups.

\section{The case of the cotangent bundle $T^*G$}
\label{S:Sec2}

Let $G$ be a connected and simply connected compact Lie group whose Lie algebra $\cG$ is simple.
In this section we describe two degenerate integrable systems on the cotangent
bundle $T^*G$ and characterize their Poisson reduction induced by
the conjugation action of $G$. The first system
contains the Hamiltonian that generates  free geodesic motion
on $G$, and its reduction leads
to a trigonometric spin Sutherland model.  The reduction of
the other system on $T^*G$  gives rational spin Ruijsenaars--Schneider type models.
Most of the results presented in this section are available in the literature \cite{FP,Res1}.
We
include their treatment
mainly in order to motivate the subsequent generalizations.
 However, the
descriptions of the reduced Poisson brackets and equations of motion
as given by Theorem \ref{thm:spinRS}
and  Proposition \ref{prop:redeq2} appear to be new.

Let us identify the dual space $\cG^*$ with $\cG$  using the (negative definite)
inner product $\langle - , - \rangle_\cG$, which
is a multiple of the Killing form, and then
identify $\cM:=T^*G$ with $G\times \cG$ using right-translations.
The canonical Poisson bracket on the phase space
\be
\cM = G \times \cG = \{(g,J)\mid g\in G,\, J\in \cG\}
\ee
can be written as
\be
\{ \cF, \cH\}(g,J) =
\langle \nabla_1 \cF, d_2 \cH\rangle_\cG - \langle \nabla_1 \cH, d_2 \cF \rangle_\cG + \langle J, [d_2 \cF, d_2 \cH]\rangle_\cG,
\label{PBcot}\ee
where the derivatives are taken at $(g,J)$.
Here and below, we use the $\cG$-valued derivatives of any $\cF \in C^\infty(\cM)$,  defined by
\be
 \langle X, \nabla_1 \cF(g,J) \rangle_\cG + \langle X', \nabla_1' \cF(g,J) \rangle_\cG := \dt \cF(e^{tX} g e^{tX'},J), \quad \forall
  (g,J)\in\cM, \, X,X' \in \cG,
 \label{nabG1}\ee
 and
 \be
 \langle X, d_2 \cF(g,J) \rangle_\cG  := \dt \cF(g, J + tX), \quad \forall
  (g,J)\in\cM, \, X \in \cG.
 \label{d2}\ee
The group $G$ acts by simultaneous conjugations of $g$ and $J$, i.e., the action of $\eta \in G$ on $\cM$
is furnished by the map
\be
A_\eta: (g,J) \mapsto (\eta g \eta^{-1}, \eta J \eta^{-1}).
\label{Aeta}\ee
This Hamiltonian action is generated by the moment map $\Phi: \cM \to \cG$,
\be
\Phi(g,J) = J - \tilde J
\quad \hbox{where}\quad \tilde J := g^{-1} J g.
\label{PhimomT}\ee
The space of $G$-invariant real functions, $C^\infty(\cM)^G$, forms a Poisson subalgebra.
By definition, this is identified  as the Poisson algebra of smooth functions carried by the quotient space $\cM/G$.

Let us first consider the invariant Hamiltonians $\cH \in C^\infty(\cM)^G$ of the form
\be
\cH(g,J) = \varphi(J)
\quad\hbox{with}\quad \varphi\in C^\infty(\cG)^G.
\label{Hams1}\ee
That is,  $\cH = \pi_2^*(\varphi)$ using the natural projection $\pi_2: \cM \to \cG$.
There are $\ell:= \rank(\cG)$ functionally independent Hamiltonians in this set,
since the ring of invariants for the adjoint action of $G$ on $\cG$,
$C^\infty(\cG)^G$, is freely generated by $\ell$ basic invariants (see, e.g., \cite{Mic}, Section 30).
The Hamiltonian vector field engendered by $\cH$ can be written as
\be
\dot{g} = (d \varphi(J)) g, \qquad \dot{J} =0,
\label{unredeq1}\ee
and its integral curve  through the initial value $(g(0), J(0))$ reads
\be
(g(t), J(t))= ( \exp(t d \varphi(J(0))) g(0), J(0)).
\label{unredsol1}\ee
The corresponding constants of motion are given by arbitrary functions of $J$ and $\tilde J$ \eqref{PhimomT}.
Since  $\psi(J) = \psi(\tilde J)$ for every function $\psi \in C^\infty(\cG)^G$, and this gives $\ell$ relations,
the functional dimension of the ring of constants of motion is $2 \dim(\cG) - \ell$.
Therefore the Hamiltonians \eqref{Hams1} form a degenerate integrable system.

 Another degenerate integrable system arises from the Hamiltonians $\cH\in C^\infty(\cM)^G$ of the form
\be
\cH(g,J) = h(g)
\quad\hbox{with}\quad h\in C^\infty(G)^G.
\label{Hams2}\ee
In other words,  $\cH = \pi_1^*(h)$ with the projection $\pi_1: \cM \to G$.
These Hamiltonians are in involution and form a ring of functional dimension $\rank(\cG)$, too.
The corresponding evolution equations read
\be
\dot{g}=0, \qquad \dot{J}  = - \nabla h(g),
\label{unredeq2}\ee
and  their flows are given by
\be
(g(t), J(t)) = \left(g(0), J(0) - t \nabla h(g(0))\right).
\label{unredsol2}\ee
The constants of motion are now found as arbitrary functions of the pair $(g,\Phi)$, where $\Phi$
is the moment map \eqref{PhimomT}.
To show that the functional dimension of this ring of functions is $2\dim(\cG) - \rank(\cG)$,
consider the isotropy subalgebra of $g$,
\be
\cG(g):= \{ X \in \cG \mid X g  - gX=0\},
\label{cGg}\ee
whose dimension  equals $\ell= \rank(\cG)$ for generic $g$.
Then notice the identity
\be
\langle \Phi(g,J), X \rangle_\cG =0\quad\hbox{for all}\quad X\in \cG(g).
\ee
On a dense open subset of $\cM$, this implies $\ell$ relations between the components of $\Phi(g,J)$,
and  apart from this  $\Phi$ varies freely if $g$ is generic.
It follows that the functional dimension of the ring of constants of motion
is $2 \dim(\cG) - \ell$, proving that the Hamiltonians \eqref{Hams2} yield a degenerate integrable system.

An element of $G$ is \emph{regular} if its isotropy group with respect to conjugations
is a maximal torus in $G$, and an  element of $\cG$ is regular if its centralizer in $\cG$
is the Lie algebra of a maximal torus.
We fix a maximal torus $G_0 < G$ and let $\cG_0$ denote its Lie algebra.
Then $G^\reg$, $G_0^\reg$ and $\cG^\reg$, $\cG_0^\reg$ stand for the corresponding
open dense subsets of regular elements.
We also introduce the following sets
\be
\cM^\reg:=\{ (g,J)\in \cM \mid g \in G^\reg\}, \qquad
\cM_0^\reg :=\{ (Q,J) \in \cM \mid Q\in G_0^\reg\},
\label{M01}\ee
and
\be
\CMpreg:=\{ (g,J) \in \cM \mid J\in \cG^\reg\}, \qquad
\cMpreg0:=\{ (g, \lambda)\in \cM \mid \lambda \in \cG_0^\reg\}.
\label{M02}\ee
The submanifolds $\cM_0^\reg  \subset \cM^\reg$ and $\cMpreg0 \subset \CMpreg$ are
stable under the action of the normalizer of $G_0$ in $G$, which we denote by $\fN$:
\be
\fN:= \{ \eta \in G \mid  \eta G_0 \eta^{-1} = G_0\}.
\label{fN}\ee
Note that $G_0$ is a normal subgroup of $\fN$, and the factor group $\fN/G_0$ is the Weyl group
of the pair $(G_0, G)$.

Any continuous function $\cF$ on $\cM$ can be recovered from its restriction to $\cM_0^\reg$,
 as well as from its restriction to ${\cM'}_0^\reg$. The restrictions of the $G$-invariant functions
enjoy residual $\fN$-invariance.
It is also easy to see  that the restrictions of functions
 provide
the following isomorphisms:
\be
C^\infty(\cM^\reg)^G \Longleftrightarrow  C^\infty(\cM_0^\reg)^\fN
\label{isom1}\ee
and
\be
C^\infty(\CMpreg)^G \Longleftrightarrow  C^\infty(\cMpreg0)^\fN.
\label{isom2}\ee

In preparation, now we introduce the dynamical $r$-matrices that will feature below.
For this purpose, we consider the decomposition
\be
\cG = \cG_0 + \cG_\perp,
\label{cGdecperp}\ee
where $\cG_\perp$ is the orthogonal complement of
 the fixed maximal Abelian subalgebra $\cG_0< \cG$ with respect to the Killing form.
Accordingly, we may write any $X\in \cG$ as
\be
X= X_0 + X_\perp \quad \hbox{where}\quad X_0\in \cG_0,\, X_\perp \in \cG_\perp.
\label{cGdecX}\ee
Then, for any $Q \in G_0^\reg$ we introduce $\cR(Q) \in \End(\cG)$ by
\be
 \cR(Q) (X) = \frac{1}{2} (\Ad_Q + \id)\circ (\Ad_Q - \id)_{\vert \cG_\perp}^{-1}(X_\perp),
 \label{RQ1}\ee
 using that $(\Ad_Q -\id)$  is invertible on $\cG_\perp$.
 Moreover, for any $\lambda \in \cG_0^\reg$ we define $r(\lambda)\in \End(\cG)$ by
 \be
 r(\lambda)( X) := (\ad_\lambda)_{\vert \cG_\perp}^{-1} (X_\perp),
 \label{rLa}\ee
 using that $\ad_\lambda$ is invertible on $\cG_\perp$.
These linear operators
are well-known solutions of the
(modified) classical dynamical Yang--Baxter equation \cite{EV}.
They vanish identically on $\cG_0$ and are antisymmetric
\be
\langle \cR(Q) X , Y \rangle_\cG = -\langle  X , \cR(Q) Y \rangle_\cG,
\qquad
\langle r(\lambda) X , Y \rangle_\cG = -\langle  X , r(\lambda) Y \rangle_\cG, \quad \forall X,Y\in \cG.
\label{Rrasym}\ee

With the necessary definitions at hand, we are ready to derive  convenient
characterizations of the Poisson algebras of the invariant functions.
We begin by noting that every $G$-invariant function on $\cM$ satisfies the basic identity
\be
g^{-1} \nabla_1 \cF(g,J) g - \nabla_1 \cF(g,J) = [J, d_2 \cF(g,J)].
\label{crucid1}\ee
This is a consequence of the property
\be
\dt \cF(e^{t X} g e^{-tX}, e^{t X} J e^{-t X})=0, \qquad \forall X \in \cG,
\ee
taking into account the equality $\nabla' \cF(g, J) = g^{-1} \nabla \cF(g,J) g$.

\subsection{Spin Sutherland models from reduction}
\label{S:Sec2.1}

For any $F\in C^\infty(\cM_0^\reg)$ define $\nabla_1 F(Q,J) \in \cG_0$ by
\be
\langle X_0, \nabla_1 F(Q,J) \rangle_\cG := \dt F(e^{t X_0} Q,J), \qquad \forall X_0\in \cG_0,\, (Q,J)\in \cM_0^\reg,
\ee
and define $d_2 F(Q,J)\in \cG$ similarly to \eqref{d2}.

\begin{thm}\label{thm:spinSuth}
Let $F,H \in C^\infty(\cM_0^\reg)^\fN$  be  the restrictions of invariant functions $\cF, \cH \in C^\infty(\cM^\reg)^G$.
Defining the reduced Poisson bracket of $F$ and $H$ by
\be
\{ F,H\}_\red(Q,J) := \{ \cF, \cH\}(Q,J), \qquad \forall (Q,J) \in \cM_0^\reg,
\label{redcot1def}\ee
the following formula holds:
\be
\{F,H\}_\red(Q,J)= \langle \nabla_1 F, d_2 H\rangle_\cG - \langle \nabla_1 H, d_2 F \rangle_\cG
+\langle J, [ \cR(Q) d_2 F, d_2 H] + [d_2 F, \cR(Q) d_2 H] \rangle_\cG,
\label{redcot1T}\ee
where $\cR(Q)$ is given by \eqref{RQ1} and the derivatives are taken at $(Q,J)$.
\end{thm}

\begin{proof}
In order to evaluate the right-hand side of \eqref{redcot1def}, we have to express the derivatives of $\cF$ and $\cH$ in terms
of the derivatives of the corresponding restricted functions. Plainly, we have
 \be
d_2 \cF(Q,J) = d_2 F(Q,J) \quad \hbox{and}\quad (\nabla_1 \cF(Q,J))_0 = \nabla_1 F(Q,J),
\ee
 where we use the decomposition \eqref{cGdecperp}. The invariance with respect to $G_0 < G$ implies $[J, d_2 F(Q,J)]_0 =0$.
For $X_\perp:= (\nabla_1 \cF(Q,J))_\perp$ and $Y_\perp:= [d_2 F(Q,J), J]$,
the identity \eqref{crucid1} gives
\be
(\Ad_{Q^{-1}} - \id) X_\perp = Y_\perp.
\ee
This can be solved:
\be
X_\perp = \Ad_{Q} \circ (\id - \Ad_Q)^{-1} Y_\perp
%= \frac{1}{2}((\Ad_{Q}+\id)+ (\Ad_Q -\id)) \circ (\id - \Ad_Q)^{-1} Y_\perp
= (\frac{1}{2}\id + \cR(Q)) Y_\perp
\ee
with $\cR(Q)$ \eqref{RQ1}, where the inverse is understood to be taken on $\cG_\perp$.
By using these equalities as well as the antisymmetry \eqref{Rrasym} and the invariance property of the Killing form,
 at $(Q,J)$ we obtain
\be
\langle \nabla_1 \cF, d_2 \cH\rangle_\cG = \langle \nabla_1 F, d_2 H \rangle_\cG +
\frac{1}{2} \langle J, [d_2 H, d_2 F]\rangle_\cG + \langle J, [d_2 F, \cR(Q) d_2 H]\rangle_\cG.
\ee
Inserting this and  $\langle \nabla_1 \cH, d_2 \cF\rangle_\cG$ into \eqref{PBcot}, together with
 $\langle J,  [d_2 \cF, d_2 \cH]\rangle_\cG = \langle J,  [d_2 F, d_2 H]\rangle_\cG$,
 leads  to  the claimed formula \eqref{redcot1T}.
\end{proof}

\begin{rem}
On account of the isomorphism \eqref{isom1},
the formula \eqref{redcot1T} defines a Poisson bracket on  $C^\infty(\cM_0^\reg)^\fN$, which
can be identified with
the ring of smooth functions on $\cM^\reg/G\simeq \cM^\reg_0/\fN$.
This space of functions is larger than $C^\infty(\cM/G)$, since there exist $G$-invariant smooth functions on $\cM^\reg$ that do not extend
smoothly to the full manifold $\cM$.
(For example \cite{FeKlu},  in the $G=\mathrm{SU(n)}$  case the ordered eigenvalues of $g\in G$ give such functions.)
On the other hand, the same formula \eqref{redcot1T} yields a Poisson bracket also on $C^\infty(\cM_0^\reg)^{G_0}$, since $\cM_0^\reg/G_0$ is
a covering space of $\cM_0^\reg/\fN$, with the fibers labeled by the elements of the Weyl group $\fN/G_0$.
To avoid any possible confusion, we note that in \eqref{redcot1T} $\langle \nabla_1 F , d_2 H\rangle_\cG = \langle \nabla_1 F , (d_2 H)_0\rangle_\cG$
since $\nabla_1 F$ is $\cG_0$-valued, and similarly for the second term.
\end{rem}

The following statement is an immediate consequence of Theorem \ref{thm:spinSuth} and the
identity
\be
[J, d\varphi(J)] =0, \quad \forall J\in \cG,
\ee
which is verified by every $\varphi \in C^\infty(\cG)^G$.

\begin{prop}\label{prop:redeq1}
If $\cH(g,J) = \varphi(J)$ with $\varphi\in C^\infty(\cG)^G$, then
for its restriction $H\in C^\infty(\cM_0^\reg)^\fN$ and any $F \in C^\infty(\cM_0^\reg)^\fN$
the Poisson bracket \eqref{redcot1T} reads
\be
\{ F,H\}_\red(Q,J) = \langle \nabla_1 F(Q,J), d \varphi(J) \rangle_\cG + \langle d_2F(Q,J), [\cR(Q) d\varphi(J), J]\rangle_\cG.
\label{redevol1}\ee
This gives the derivative of $F$ with respect to an evolution vector field on $\cM_0^\reg$,  and
the corresponding  `reduced evolution equation' on $\cM_0^\reg$ can be taken to be
\be
\dot{Q} = (d\varphi(J))_0 Q, \qquad \dot{J} = [\cR(Q) d\varphi(J), J].
\label{redeq1}\ee
\end{prop}

The solutions of the evolution equation \eqref{redeq1} result by applying suitable (point dependent) $G$-transformations
to the unreduced integral curves \eqref{unredsol1}, and they project onto the reduced dynamics on $\cM^\reg/G \simeq \cM_0^\reg/\fN$.
This follows from the general theory of Hamiltonian reduction \cite{OR}.
Of course, the evolution vector field on $\cM_0^\reg$ is not unique, because the derivative of
$F\in C^\infty(\cM_0^\reg)^\fN$ is
zero along any vector field that is tangent to the orbits of $G_0$ in $\cM_0^\reg$.
We fixed this ambiguity by requiring that the derivative of \emph{any} $F\in C^\infty(\cM_0^\reg)$ should be given
by the right-hand side of  \eqref{redevol1}.

The reduced system governed by the Poisson bracket \eqref{redcot1T} and equations of motion \eqref{redeq1} can
be interpreted as a spin Sutherland model. Since this is well known \cite{FP,LX}, we only note that for
$\varphi(J):= -\frac{1}{2} \langle J, J\rangle_\cG$
the parametrization \eqref{Jpar}  of $J$ by the new variables $ q$ (with $Q=e^{\ri q}$), $p$ and $\xi$
leads to
\be
-\frac{1}{2} \langle J, J\rangle_\cG = -\frac{1}{2} \langle \ri p, \ri p\rangle_\cG
+  \frac{1}{2} \sum_{\alpha>0} \frac{ \vert \xi_\alpha \vert^2}{\vert \alpha \vert^2 \sin^2(\alpha(q)/2)},
\label{spinSuthHam}\ee
which is a standard spin Sutherland Hamiltonian.
Here, the sum is over the positive roots of the complexification $\cG^\bC$ of $\cG$ with respect
to the Cartan subalgebra $\cG_0^\bC<\cG^\bC$, and the spin variable $\xi\in \cG_\perp$ is expanded as
$\xi = \sum_{\alpha>0} ( \xi_\alpha E_\alpha - \xi_\alpha^* E_{-\alpha})$ using  root vectors
$E_{\pm \alpha}$ (normalized according to Appendix \ref{sec:A}).

\subsubsection{Degenerate integrability after reduction.}
We now discuss how the mechanism outlined around equation \eqref{Psired} is applicable to the present case.
By inspecting the restriction on $\cM_0^\reg$,
it is easily seen that  $\pi_2^*(C^\infty(\cG)^G)$ gives rise to $\ell=\rank(\cG)$ generically independent
Hamiltonians on $\cM/G$.
Let us now define the smooth,  $G$-equivariant map $\Psi_1: \cM \to \cG \times \cG$ by
\be
\Psi_1(g,J) := (\tilde J, J) \quad\hbox{with}\quad \tilde J = g^{-1} J g,
\label{Psi1}\ee
where $\eta \in G$ acts on $\cG \times \cG$ by applying $\Ad_\eta$ to  both components of $(a,b)\in \cG\times \cG$.
Then, taking any function $\chi\in C^\infty(\cG \times \cG)^G$, the function $\Psi_1^* (\chi)\in C^\infty(\cM)^G$
is a smooth, $G$-invariant constant of motion.
We next outline a train of thought indicating that
these constants of motion guarantee degenerate integrability after reduction.

The isotropy subgroup of generic elements from the image of $\Psi_1$
is clearly just the center $Z_G$ of $G$.  These generic elements form a manifold $Y$ of dimension
$2 \dim(G) - \ell$, and its pre-image $X \subset \cM$ is a dense, open, $G$-invariant  subset.
Thus, taking $\Psi:=\Psi_1$  in \eqref{Psired},
we obtain $\dim(G) - \ell$ functionally independent constants of motion for
the restriction of the reduced system to $X/G \subset \cM/G$.
By using the moment map $\Phi$ \eqref{PhimomT},
the $G$-invariant functions of the form $\phi \circ \Phi$, with any $\phi \in C^\infty(\cG)^G$,
descend to $\ell$ independent \emph{Casimir functions} on $\cM/G$.
Fixing the values of these Casimir functions, generically one obtains a symplectic leaf of dimension $\dim(G) - \ell$ in $\cM/G$.
Thus, on the intersection of such a generic symplectic leaf with $X/G$, there
 remain $\dim(G) - 2 \ell$ independent constants of motion.
This is sufficient
for degenerate integrability since the commuting reduced Hamiltonians remain independent
on the generic symplectic leaves.

The above arguments
make us confident to expect degenerate integrability on generic
symplectic leaves of $\cM/G$.
These arguments essentially coincide with those presented by Reshetikin \cite{Res1} for the
corresponding complex holomorphic systems.
A more complete, rigorous analysis of reduced integrability is beyond the scope of the present paper.

\subsection{The duals of the spin Sutherland models}

Now we turn to the characterization of the Poisson algebra of the invariant
functions in terms of their restriction to ${\cM'}_0^\reg$
\eqref{M02}.
For any $F\in C^\infty({\cMpreg0})$,  $d_2 F(g,\lambda) \in \cG_0$ is defined by
\be
\langle X_0, d_2 F(g,\lambda) \rangle_\cG := \dt F(g ,\lambda + t X_0), \qquad
\forall X_0\in \cG_0,\, (g,\lambda)\in {\cM'}_0^\reg,
\ee
and the derivatives with respect to the first variable are given by \eqref{nabG1}.

\begin{thm}\label{thm:spinRS}
Let $F,H \in C^\infty({\cMpreg0})^\fN$  be  restrictions of invariant functions $\cF, \cH \in C^\infty(\CMpreg)^G$.
Defining the reduced Poisson brackets of $F$ and $H$ by
\be
\{ F,H\}'_\red(g,\lambda) := \{ \cF, \cH\}(g,\lambda), \qquad \forall (g,\lambda) \in \cMpreg0,
\label{redprime1}\ee
the following formula holds:
\be
 \{F,H\}'_\red(g,\lambda) =
 \langle \nabla_1 F, d_2 H \rangle_\cG - \langle \nabla_1 H, d_2 F \rangle_\cG
+\langle \nabla_1' F, r(\lambda) \nabla_1' H \rangle_\cG   - \langle \nabla_1 F, r(\lambda) \nabla_1 H \rangle_\cG,
\label{redcot2}\ee
where $r(\lambda)$ is given by \eqref{rLa} and the derivatives are taken at $(g,\lambda)$.
\end{thm}

\begin{proof}
First of all, we remark that
\be
\nabla_1 \cF(g,\lambda) = \nabla_1 F(g, \lambda),\quad (d_2 \cF(g,\lambda))_0 = d_2 F(g,\lambda),
\ee
and, as a consequence of the invariance under $G_0< G$,
\be
\left( \nabla_1' \cF(g,\lambda) - \nabla_1 \cF(g,\lambda)\right)_0 = 0.
\ee
The subscript $0$ refers to  the decomposition \eqref{cGdecX}.
In view of the formula \eqref{PBcot}, we have to express $(d_2 \cF(g,\lambda))_\perp$  in terms of the
derivatives of $F$ with respect to the variable $g$.
By applying \eqref{crucid1} at $(g,\lambda)$ and using the above relations,
we find
\be
(d_2 \cF (g,\lambda))_\perp = r(\lambda) \left( \nabla_1' F(g,\lambda) - \nabla_1 F(g,\lambda)\right)
\ee
with $r(\lambda)$  \eqref{rLa}.
Then, substitution in the right-hand side of \eqref{redprime1} leads to
\be
\begin{aligned}
\{ F,H\}_\red' &= \langle \nabla_1 F , d_2 H\rangle_\cG - \langle \nabla_1 H , d_2 F  \rangle_\cG\\
 +& \langle \nabla_1 F, r(\lambda) (\nabla_1' H - \nabla_1 H)\rangle_\cG - \langle \nabla_1 H, r(\lambda) (\nabla_1' F - \nabla_1 F)\rangle_\cG\\
 + &  \langle \lambda, [r(\lambda) (\nabla_1' H - \nabla_1 H),  r(\lambda) (\nabla_1' F - \nabla_1 F)]\rangle_\cG .
\end{aligned}
\label{intermed}\ee
This can be simplified by virtue of the classical dynamical Yang--Baxter equation \cite{EV}, which can be written as \cite{FP}
\be
[r(\lambda) X, r(\lambda) Y]  = r(\lambda)\left( [ X, r(\lambda) Y] + [ r(\lambda) X, Y]\right)
  +  d_{Y_0} r(\lambda)X - d_{X_0} r(\lambda) Y + \sum_{i}K^i \langle X, d_{K_i} r(\lambda) Y\rangle_\cG,
\label{CDYBE}\ee
$\forall X,Y\in \cG$, where $d_{X_0}, d_{Y_0}$ and $d_{K_i}$ are directional derivatives, and $\langle K_i, K^j\rangle = \delta_i^j $
with a  pair of dual bases of $\cG_0$.
We observe that
\be
\langle \lambda, \sum_{i} K^i \langle X, d_{K_i} r(\lambda) Y\rangle_\cG =
 \langle X, d_{\lambda} r(\lambda) Y\rangle
 \quad\hbox{and}\quad  d_{\lambda} r(\lambda) = - r(\lambda).
 \ee
We now take $X:=(\nabla_1' H - \nabla_1 H)$, and  $Y:= (\nabla_1' F - \nabla_1 F)$,
for which $X_0 = Y_0 =0$.   Noticing  that
$\langle \lambda, r(\lambda) Z \rangle_\cG = 0$ for all $Z\in \cG$, because $r(\lambda)$ is antisymmetric \eqref{Rrasym} and vanishes on $\cG_0$,
we obtain
\be
 \langle \lambda, [r(\lambda) (\nabla_1' H - \nabla_1 H),  r(\lambda) (\nabla_1' F - \nabla_1 F)]\rangle_\cG =
 - \langle \nabla_1' H - \nabla_1 H, r(\lambda) (\nabla_1' F - \nabla_1 F) \rangle_\cG.
 \ee
 By inserting this into \eqref{intermed}  and collecting terms, we arrive at the claimed formula \eqref{redcot2}.
\end{proof}

\begin{rem}
The formula \eqref{redcot2} defines  a Poisson bracket not only on $C^\infty(\cMpreg0)^\fN$, but also on $C^\infty(\cMpreg0)^{G_0}$.
It should be noted that  $\langle \nabla_1 F , d_2 H\rangle_\cG = \langle (\nabla_1 F)_0 , d_2 H\rangle_\cG$
since $d_2H$ is $\cG_0$-valued.
\end{rem}

The next result follows from Theorem \ref{thm:spinRS} by using that
\be
\nabla h = \nabla' h, \quad \forall h\in C^\infty(G)^G.
\ee

\begin{prop}\label{prop:redeq2}
If $H$ is the restriction of an invariant Hamiltonian $\cH= \pi_2^*h$ displayed in \eqref{Hams2},
then \eqref{redcot2} simplifies to
\be
 \{F,H\}'_\red(g,\lambda) =
  - \langle  d_2 F(g, \lambda), \nabla h(g) \rangle_\cG
+\langle \nabla_1' F(g,\lambda) - \nabla_1 F(g,\lambda), r(\lambda) \nabla h(g) \rangle_\cG.
\label{redcot2+}\ee
The corresponding
reduced evolution equation on $\cMpreg0$  can be taken to be
\be
\dot{g} = [g, r(\lambda) \nabla h(g)], \qquad \dot{\lambda} = -(\nabla h(g))_0.
\label{redeq2}\ee
\end{prop}

The counterpart of the  discussion presented after Proposition \ref{prop:redeq1} is applicable
in this case as well. We merely note that the solutions of the evolution equations \eqref{redeq2} can be obtained
 by applying
suitable $G$-transformations to those unreduced integral curves \eqref{unredsol2},
whose initial values belong to $\cMpreg0$.

It is known \cite{FA} that in the $G=\mathrm{SU}(n)$ case the above reduced system contains
a real form of the rational Ruijsenaars--Schneider model
on a special symplectic leaf.
The leaf in question arises by fixing the Casimir functions $\phi \circ \Phi$ ($\phi \in C^\infty(\cG)^G$)
in such a way that the corresponding joint level surface in $\cG\simeq \cG^*$ is a minimal (co)adjoint orbit of dimension $2(n-1)$.
 The main Hamiltonian of this model
is associated with the function $h(g) = \Re \tr(g)$ on $G$.
This lends justification to the terminology `spin Ruijsenaars--Schneider type models' \cite{Res1}
as a name for the  models that stem from the integrable Hamiltonians \eqref{Hams2} in general.
However, in contrast to the spin Sutherland models
described in the preceding subsection, it is still an open problem to separate the variables of these models into canonically
conjugate
pairs complementing the components of $\lambda$  and additional `spin' degrees of freedom.

\subsubsection{Degenerate integrability and duality.}
The degenerate integrability of the reduced systems built on the pullback invariants
$\pi_1^*(C^\infty(G)^G)$ can be analyzed quite similarly to the previous case of $\pi_2^*(C^\infty(\cG)^G)$.
Now one may use the map
\be
\Psi_2: \cM \to G\times \cG
\quad\hbox{defined by}\quad \Psi_2(g,J):= (g, \Phi(g,J)),
\label{Psi2}\ee
 which is constant along the flows of any $\cH\in \pi_1^*(C^\infty(G)^G)$, and is $G$-equivariant
 with respect to the same action that operates on $\cM$.
 The arguments presented at the end of Subsection 2.1 go through with little modification, as is
discussed in \cite{Res1} in the holomorphic case.
In particular, employing any  $\chi\in C^\infty(G \times \cG)^G$,
the function $\Psi_2^* (\chi)\in C^\infty(\cM)^G$
is a smooth, $G$-invariant constant of motion.

Incidentally, the maps $\Psi_1$ \eqref{Psi1} and $\Psi_2$ \eqref{Psi2} are Poisson maps with respect to suitable
Poisson structures on the target spaces $\cG \times \cG$ and $G\times \cG$, which can
be easily found by requiring this property to hold.
Therefore the just mentioned $G$-invariant constants of motion $\Psi_2^* (\chi)$  form a closed Poisson subalgebra of $C^\infty(\cM)^G$
(and similarly for $\Psi_1$).

Finally, let us comment on the duality between the spin Sutherland and the spin Ruijsenaars--Schneider systems.
To this end, we regard the functions of $Q$ in \eqref{M01} and $\lambda$ in \eqref{M02}
as `position variables' for the respective models. Those functions of $Q$ that descend to well-defined
functions on $\cM/G$ arise from $\pi_1^*(C^\infty(G)^G)$ and the functions of $\lambda$ having the same property
arise from $\pi_2^*(C^\infty(\cG)^G)$.
In this way, one of the two sets of pullback invariants plays the role of `global position variables'
in every reduced system, while the other set engenders the commuting Hamiltonians of interest of the same system.
The role of the two sets of pullback invariant is interchanged in the two systems. That is, since
both systems leave on the same phase space $\cM/G$, the global position variables of one system are the
interesting Hamiltonians of the other one, and vice versa.
This kind of duality was originally discovered by Ruijsenaars for spinless models (see the review \cite{RBanff} and references therein).
We call it \emph{Ruijsenaars duality} or \emph{action-position duality}, taking into
account that in integrable models the commuting Hamiltonians are in bijective correspondence with the action variables.
We prefer this to the term action-angle duality, which is also used in the literature.

\section{Integrable systems from the Heisenberg double}
\label{S:Sec3}

In this section we first describe the Heisenberg double associated with a compact
Lie group $G$,  and specify two degenerate integrable systems on this phase space.
We then study the Poisson reduction of these systems.
For notations, see the remark at the end of Section \ref{S:Sec1}, and also Appendix \ref{sec:A}.
For the underlying theory of Poisson--Lie groups,
one may consult the reviews \cite{KS,STS+}.

\subsection{The basics of the Heisenberg double}

We start with a compact simple Lie algebra, $\cG$, and pick a maximal Abelian subalgebra, $\cG_0$.
These can be regarded as real forms of a complex simple Lie algebra, $\cG^\bC$, and its Cartan
subalgebra, $\cG_0^\bC$. Choosing a system of positive roots, we obtain the triangular
decomposition
\be
\cG^\bC = \cG_<^\bC + \cG_0^\bC + \cG_>^\bC,
\label{triang}\ee
where $\cG^\bC_>$ is spanned by the eigenvectors associated with the positive roots.
Referring to this, we may present any $X\in \cG^\bC$ as
\be
X = X_< + X_0 +  X_>
\ee
with the terms taken from the corresponding subspaces.
The real vector space
\be
\cB := \ri \cG_0 + \cG^\bC_>
\label{cBtriang}\ee
is a Lie subalgebra of the `realification' $\cG^\bC_\bR$ of $\cG^\bC$ (i.e. $\cG^\bC$ viewed as a real Lie algebra),
and it gives rise to the direct sum decomposition
\be
\cG^\bC_\bR = \cG + \cB.
\label{cGdec}\ee
Correspondingly, we may write any $X\in \cG^\bC_\bR$ as
\be
X= X_\cG + X_\cB, \qquad X_\cG\in \cG,\, X_\cB\in \cB.
\label{cGdec+}\ee
We equip $\cG^\bC_\bR$ with the invariant, non-degenerate, symmetric bilinear form $\langle - , - \rangle_\bI$,
 defined as the imaginary part of the
 complex Killing form $\langle - , - \rangle$ of $\cG^\bC$.
The decomposition \eqref{cGdec} represents a so-called Manin triple \cite{KS,STS+}, since $\cG$ and $\cB$ are isotropic subalgebras of $\cG_\bR^\bC$
with respect to $\langle - , - \rangle_\bI$.

Let $G^\bC_\bR$
be a connected and simply connected real Lie group whose Lie algebra is $\cG^\bC_\bR$, and denote $G$ and $B$ its
connected subgroups associated with the Lie subalgebras $\cG$ and $\cB$.
These subgroups are simply connected and $G$ is compact.
Later we shall also need the connected subgroup $G_0^\bC <  G^\bC_\bR$ associated with $\cG_0^\bC$
as well as the subgroups $G_0 < G$ and $B_0 < B$ associated with $\cG_0$ and $\ri\cG_0$.
Occasionally, we view $G^\bC_\bR$ as the realification of the corresponding complex Lie group, $G^\bC$.

Now, we recall \cite{STS,STS+} that the group manifold
\be
M:= G^\bC_\bR
\label{MHeis}\ee
 carries the following two natural Poisson brackets:
\be
\{ F, H\}_{\pm}: = \langle \nabla F, \rho \nabla H \rangle_\bI \pm  \langle \nabla' F, \rho \nabla' H \rangle_\bI
\quad\hbox{with}\quad \rho := \frac{1}{2}\left( \pi_{\cG} - \pi_{\cB}\right),
\label{A1T}\ee
where $\pi_\cG$ and $\pi_\cB$  are the projections from $\cG^\bC_\bR$ onto $\cG$ and $\cB$, respectively,
defined by means of \eqref{cGdec}.
Here, we use the $\cG^\bC_\bR$-valued  `left- and right-derivatives' of $F,H \in C^\infty(M)$:
 \be
 \langle X, \nabla F(K) \rangle_\bI + \langle X', \nabla' F(K) \rangle_\bI := \dt F(e^{tX} K e^{tX'}), \quad \forall K\in M, \, X,X' \in \cG_\bR^\bC.
 \label{Nab}\ee
 The minus bracket makes $M$ into a Poisson--Lie group, of which $G$ and $B$ are Poisson--Lie subgroups.
 Their inherited Poisson brackets take the form
\be
\{ \varphi_1, \varphi_2\}_B(b) = \langle D' \varphi_1(b), b^{-1} (D \varphi_2(b)) b \rangle_\bI,
\label{PBBT}\ee
and
\be
\{ f_1, f_2\}_G(g) = - \langle D' f_1(g), g^{-1} (D f_2(g)) g \rangle_\bI.
\label{PBGT}\ee
The derivatives are $\cG$-valued for  $\varphi_i\in C^\infty(B)$ and $\cB$-valued for $f_i\in C^\infty(G)$,
reflecting that these subalgebras are in duality with respect to $\langle -,- \rangle_\bI$.
To be sure, we write the definitions
\be
 \langle X, D \varphi(b) \rangle_\bI + \langle X', D' \varphi(b) \rangle_\bI := \dt \varphi(e^{tX} b e^{tX'}), \quad \forall b\in B, \, X,X' \in \cB,
 \label{derB}\ee
\be
 \langle X, D f(g) \rangle_\bI + \langle X', D' f(g) \rangle_\bI := \dt f(e^{tX} g e^{tX'}), \quad \forall g\in G, \, X,X' \in \cG,
\label{derG} \ee
where $\varphi \in C^\infty(B)$ and $f\in C^\infty(G)$. We shall also use the $\cG$-valued derivatives of $f\in C^\infty(G)$,
\be
 \langle X, \nabla f(g) \rangle_\cG + \langle X', \nabla' f(g) \rangle_\cG := \dt f(e^{tX} g e^{tX'}), \quad \forall g\in G, \, X,X' \in \cG,
 \label{nabG}\ee
and note that the Killing form $\langle -,- \rangle_\cG$ of $\cG$ is the
the restriction to $\cG$ of the complex Killing form $\langle -,- \rangle$ of $\cG^\bC$.
One has
\be
\langle X, Y\rangle_\cG = \langle X,\ri Y\rangle_\bI = \langle X, (\ri Y)_\cB\rangle_\bI, \quad \forall X,Y\in \cG,
\ee
and thus the two kinds of derivatives of $f\in C^\infty(G)$ are related by
\be
D f = (\ri \nabla f)_\cB.
\label{Dfromnab}\ee
Defining $R^\ri \in \End(\cG)$ by
\be
R^\ri (X) := (- \ri X)_\cG, \quad \forall X\in \cG,
\label{Ri}\ee
 the relation of the derivatives can also be written as
\be
D f= \ri \nabla f + R^\ri (\nabla f).
\label{Dnab}\ee
Of course, analogous relations hold for the right-derivative $D'f$, too.
With these relations at hand, one can prove the identity
\be
-\langle D' f_1(g), g^{-1} D f_2(g) g \rangle_\bI = \langle \nabla' f_1(g), R^\ri \nabla' f_2(g) \rangle_\cG
- \langle \nabla f_1(g), R^\ri \nabla f_2(g) \rangle_\cG.
\label{GPL}\ee
In terms of the decomposition
$ X = X_> + X_0 + X_<$,  one has
\be
R^\ri (X) = \ri (X_> - X_<),
\ee
and the right-hand side of \eqref{GPL} has the familiar form of a Sklyanin bracket.

The Poisson bracket $\{- ,- \}_+$ \eqref{A1T} corresponds to a symplectic form \cite{AM}, and $(M, \{- ,- \}_+)$ is called \cite{STS} the Heisenberg double of
the Poisson--Lie groups $G$ and $B$. It is a Poisson--Lie analog\footnote{The cotangent bundle of any Lie group can
be viewed as the Heisenberg double of the group equipped with the identically zero Poisson bracket \cite{STS}.}
 of the cotangent bundle $T^*G$ (and of $T^*B$).
Any element $K\in M$ admits unique (Iwasawa) decompositions \cite{Knapp}  into products of elements of $G$ and $B$, which we write as
\be
K = g_L b_R^{-1} = b_L g_R^{-1} \quad \hbox{with}\quad g_L, g_R \in G,\, b_L, b_R \in B.
\label{KdecT}\ee
These decompositions give rise to the maps $\Xi_L, \Xi_R: M\to G$ and $\Lambda_L, \Lambda_R: M\to B$,
\be
\Xi_L(K) := g_L,\quad \Xi_R(K):= g_R,\quad \Lambda_L(K):= b_L,\quad \Lambda_R(K):= b_R.
\label{XiLaT}
\ee
These are all Poisson maps from the $(M, \{-,- \}_+)$ onto the respective Poisson--Lie groups, and
the same is true for the products of any two of these maps into the same group.
Without going into details, we recall that any Poisson map into a Poisson--Lie group serves as a moment map
that generates a (possibly only infinitesimal)
Poisson--Lie action of the corresponding dual group \cite{Lu}.  In particular, the Poisson map $\Lambda: M\to B$ defined by
\be
\Lambda(K):= \Lambda_L(K) \Lambda_R(K), \quad \forall K\in M,
\label{mom}\ee
generates the so-called quasi-adjoint action of $G$ on the Heisenberg double. As was shown in \cite{Kli}, the corresponding action map,
$\cA^1: G \times M \to M$, is given by
\be
\cA^1(\eta, K) = \eta K \Xi_R( \eta \Lambda_L(K)),
\qquad
\forall \eta \in G,\, K\in M.
\label{cA1}\ee
This is a Poisson map if $G\times M$ is equipped
with the product Poisson bracket coming from  $(G,\{- ,- \}_G)$ and $(M,\{-,- \}_+)$.
According to the general theory \cite{STS}, the ring of $G$-invariant real functions, $C^\infty(M)^G$, forms a Poisson
subalgebra of $(C^\infty(M), \{- ,- \}_+)$, which is, by definition, the Poisson algebra of smooth functions
on the quotient space $M/G$.
Taking this quotient is  an example of Poisson reduction.
It is worth noting that $C^\infty(M)^G$ is nothing but the centralizer of
$\Lambda^* C^\infty(B)$, i.e., $C^\infty(M)^G$  consists of those functions that Poisson commute with the
functions depending only on the moment map $\Lambda$ \eqref{mom}.

For the implementation of the Poisson reduction, an alternative  model of the Heisenberg double will also prove convenient.
This model,  which is akin to a
trivialization of the cotangent bundle $T^*G$,  is the manifold
\be
\fM:= G \times B,
\label{fM}\ee
and we transfer the Poisson bracket from $M$ to $\fM$ by means of the diffeomorphism
\be
m: M\to \fM \quad\hbox{defined by} \quad
m:= (\Xi_R, \Lambda_R).
\label{mMfM}\ee
Said more directly,  the pair $(g_R, b_R)= m(K)$ is used as a new variable instead of $K\in M\equiv G_\bR^\bC$.
It is shown in Appendix \ref{sec:B} that the map $m$ is a Poisson diffeomorphism if  $\fM$ is endowed with the following
Poisson bracket:
\be
\{\cF, \cH\}(g,b) =\left\langle D_2' \cF, b^{-1} (D_2\cH) b \right\rangle_\bI
-\left\langle D'_1\cF, g^{-1} (D_1\cH) g\right\rangle_\bI
 +  \left\langle D_1\cF , D_2\cH \right\rangle_\bI
-\left\langle D_1 \cH , D_2\cF \right\rangle_\bI
\label{A2T}\ee
for functions $\cF, \cH\in C^\infty(\fM)$.
The
 derivatives on the right-hand side
are   taken at $(g,b)\in G\times B$, with respect to the first and second variable,
according to the definitions \eqref{derG} and \eqref{derB}, respectively.
In particular, $D_1 \cF$ is $\cB$-valued and $D_2 \cF$ is $\cG$-valued.
An alternative form of \eqref{A2T} results
by employing $\cG$-valued derivatives with respect to the first variable, defined like in \eqref{nabG}.

In terms of the model $\fM$,  the quasi-adjoint action $\cA^1$ \eqref{cA1} turns into
$\cA^2: G \times \fM \to \fM$,
\be
\cA^2(\eta, (g,b)) =  \left( \Xi_R(\eta b_L)^{-1} g \Xi_R(\eta b_L), \Dress_{\Xi_R(\eta b_L)^{-1}}(b) \right),
\label{cA2}\ee
where $b_L = \Lambda_L(g^{-1} b)^{-1}$. (This expression of $b_L=\Lambda_L(K)$ in terms on
$(g,b)\equiv (g_R, b_R)$
is obtained from \eqref{KdecT} with \eqref{XiLaT}.)
Here, we use the dressing action of $G$ on $B$, defined by
\be
\Dress_\eta (b) := \Lambda_L(\eta b),
\quad
\forall \eta \in G,\, b\in B,
\label{Dress}\ee
whose infinitesimal version reads
\be
\dress_X(b):= \dt \Dress_{e^{tX}} (b) = b (b^{-1} X b)_\cB,\qquad \forall X\in \cG,
\label{dressT}\ee
where the decomposition \eqref{cGdec+} is applied to $(b^{-1} X b)\in \cG^\bC_\bR$.
The action $\cA^2$  is related to $\cA^1$ according to
\be
\cA^2_\eta \circ m  = m\circ \cA^1_\eta, \quad \forall \eta \in G,
\label{cA12}\ee
where $\cA^i_\eta$ denotes the map of the relevant manifold obtained by fixing the first argument of $\cA^i$.
We observe that
the $G$-action $\cA^2$ \eqref{cA2} has the same orbits as the simpler
action given by the map
$\cA: G \times \fM \to \fM$:
\be
\cA(\eta, (g,b)): = (\eta g \eta^{-1}, \Dress_\eta(b)).
\label{cA}\ee
Since the orbits of $\cA$  are the same as those of the Poisson--Lie action $\cA^2$, these two $G$-actions share
the same invariant functions, and thus are equivalent from the point of view of Poisson reduction.

The real Lie algebra $\cG_\bR^\bC$ carries the Cartan involution, $\theta$,  that fixes $\cG$ pointwise and multiplies the elements of $\ri \cG$ by $-1$.
It lifts to a corresponding involutive automorphism  $\Theta$ of $G_\bR^\bC$, of which $G < G_\bR^\bC$
is the fixed point set.
Referring to \eqref{triang},  $\theta$ maps $\cG^\bC_>$ onto $\cG^\bC_<$.
We shall use the notations
\be
Z^\tau:= - \theta(Z), \qquad
K^\tau:= \Theta(K^{-1}),
\qquad\forall Z\in \cG_\bR^\bC,\,\, \forall K\in G_\bR^\bC.
\label{taumap}\ee
The maps $Z \mapsto Z^\tau$ and $K \mapsto K^\tau$ are anti-automorphisms
 satisfying
\be
X^\tau = - X,\quad \forall X\in \cG \quad \hbox{and}\quad K^\tau = K^{-1},\quad \forall K\in G.
\ee
This operation is often denoted simply by dagger, since for the classical Lie groups one can
choose the conventions in such a way that $X^\tau = X^\dagger$
and $K^\tau = K^\dagger$ with dagger denoting the matrix adjoint \cite{Knapp}.
Later we shall also need the closed submanifold
\be
\fP:= \exp(\ri \cG) \subset G^\bC_\bR,
\label{fP}\ee
which is diffeomorphic not only to $\cG$ but also to $B$.
Note that $\fP$ is a connected component of the fixed point set of the anti-automorphism $K\mapsto K^\tau$
of $G_\bR^\bC$, and
 a diffeomorphism with $B$ is provided  by the map
 \be
 \nu: B \to \fP, \quad \nu(b):= b b^\tau.
\label{nu}\ee
The map \eqref{nu} intertwines the dressing action with the obvious conjugation action of $G$ on $\fP$, since we have
\be
\Dress_\eta(b) (\Dress_\eta(b))^\tau = \eta b b^\tau \eta^{-1}, \qquad \forall \eta \in G,\, b\in B.
\label{Dressconj}\ee
This implies that any element of $B$ can be transformed into $B_0= \exp(\ri \cG_0)$ by the dressing action.
As an alternative to $G\times B$, one may also take $G\times \fP$ as a model of the Heisenberg double.

\begin{rem}
After small notational changes, all considerations of the paper apply to reductive compact Lie groups as well.
For example, one can take $G= \mathrm{U}(n)$, $\cG^\bC =\mathrm{ gl}(n,\bC)$, and $G^\bC =\mathrm{ GL}(n,\bC)$, in which case $B$ can be taken
to be the upper triangular subgroup whose diagonal elements  are positive real numbers. Then, $K^\tau = K^\dagger$,
and $\fP$ is the space of positive definite, Hermitian matrices.
The reader may  keep this example (or the example of $G=\mathrm{SU}(n)$) in mind when reading the text.
We restricted ourselves to simple Lie algebras just in order have a shorter presentation.
\end{rem}

\subsection{Two degenerate integrable systems on the Heisenberg double}
Now we present two degenerate integrable systems.
For this, we let $\pi_1$ and $\pi_2$ be the projections from $\fM$ onto $G$ and $B$, respectively,
\be
\pi_1: (g,b) \mapsto g, \quad \pi_2: (g,b) \mapsto b.
\ee
Then, consider the following families of functions on $\fM$,
\be
\pi_1^*(C^\infty(G)^G)
\quad \hbox{and}\quad \pi_2^* (C^\infty(B)^G),
\label{fMHams}\ee
where the superscript refers to invariance with respect to the conjugation and dressing actions of $G$ on $G$ and on $B$,
respectively.
When presented in terms of the model $M$, these become
\be
\Xi_R^*(C^\infty(G)^G) \quad \hbox{and} \quad  \Lambda^*_R (C^\infty(B)^G),
\label{MHams}\ee
since $\pi_1\circ m = \Xi_R$ and $\pi_2 \circ m = \Lambda_R$.
Both of these rings of functions have functional dimension $\ell= \rank(\cG) \equiv \dim(\cG_0)$, since this true
for $C^\infty(G)^G$ and for $C^\infty(B)^G$ (see Appendix A).
All the Hamiltonians in \eqref{MHams} are invariant under the quasi-adjoint action of $G$ on $M$,
as is easily seen from equations \eqref{cA2} and \eqref{cA12}.
In order to see that they yield two Abelian Poisson algebras and to identify
their constants of motion, let us describe the flows generated by these Hamiltonians.
For this, we notice from \eqref{A1T} that the Hamiltonian vector field belonging to  $H\in C^\infty(M)$
generates the evolution equation
\be
\dot{K} = \rho( \nabla H(K)) K + K \rho(\nabla' H(K)).
\label{HV}\ee

\begin{prop}\label{prop:sol1}
The Hamiltonian
$H = \Lambda_R^* \phi$ with $\phi \in C^\infty(B)^G$ generates the following
evolution equation on the Heisenberg double $M=G^\bC_\bR$  by means of the Poisson bracket $\{-,- \}_+$ \eqref{A1T},
\be
\dot{K} =- K D \phi(b_R),
\label{S1}\ee
which in terms of the decompositions $K = b_L g_R^{-1} = g_L b_R^{-1}$ \eqref{KdecT} gives
\be
\dot{g}_R = D \phi(b_R) g_R, \quad  \dot{b}_L= \dot{b}_R=0, \quad \dot{g}_L = - g_L D'\phi(b_R).
\label{S2}\ee
The solution $K(t)$ corresponding
to the initial value $K(0)$ is provided by
\be
K(t) = K(0) \exp\left( - t D\phi(b_R(0))\right),
\label{S3}\ee
or equivalently
\be
b_R(t) = b_R(0),\quad b_L(t) = b_L(0), \quad g_R(t) =  \exp(  t D\phi(b_R(0))) g_R(0), \quad
 g_L(t) = g_L(0) \exp( - t D'\phi(b_R(0))).
\label{S4} \ee
\end{prop}
\begin{proof}
We begin by pointing out that $\phi\in C^\infty(B)^G$ satisfies
\be
D\phi(b) = b D'\phi(b) b^{-1}, \qquad \forall b\in B.
\label{S6}\ee
For arbitrary $\phi\in C^\infty(B)$ one has $D\phi(b)=(b D' \phi(b) b^{-1})_\cG$.
By \eqref{dressT}, the infinitesimal dressing invariance means that $\langle D'\phi(b), (b^{-1} X b)_\cB \rangle_\bI =0$ for all $X\in \cG$.
This is equivalent to $(b D'\phi(b) b^{-1})_\cB=0$, which implies \eqref{S6}.
By using Lemma \ref{lm:B2} in Appendix \ref{sec:B}, we then get
\be
\nabla' H(K) = - D\phi(b_R), \qquad \nabla H(K) = - g_L D' \phi(b_R) g_L^{-1},
\ee
which are both $\cG$-valued.
Formula \eqref{S1} follows by putting these derivatives into \eqref{HV},
\be
\dot{K} = -\frac{1}{2} ( K D\phi(b_R) + g_L D' \phi(b_R) g_L^{-1} K) = - K D \phi(b_R),
\ee
where we applied \eqref{S6} and the decomposition $K= g_L b_R^{-1}$.
By taking $K = b_L g_R^{-1}$ and using that $D\phi(b_R)$ is $\cG$-valued, \eqref{S1} implies
$\dot{g}_R = D \phi(b_R) g_R$ and $\dot{b}_L=0$.
It follows that $b_L$ remains constant. The moment map $\Lambda$ \eqref{mom} is also constant along the flow,
for $H\in C^\infty(M)^G$, and therefore $b_R$ stays constant as well.
Hence we obtain the formula for $g_R(t)$.

The formula for the time development of $g_L$ then follows directly from \eqref{S2}, or alternatively from the identity $g_L = b_L  g_R^{-1} b_R$.
In detail,
\be
\begin{aligned}
g_L(t)& = b_L(0)  g_R(0)^{-1}  \exp( - t D\phi(b_R(0)))  b_R(0)  \\
& =  g_L(0) b_R(0)^{-1} \exp( - t D\phi(b_R(0)))  b_R(0)
= g_L(0) \exp( - t D'\phi(b_R(0))),
\end{aligned}
\ee
where we took \eqref{S6} into account. This also provides a consistency check on our calculations.
\end{proof}

Since all smooth functions depending on $b_L$ and $b_R$ are constants
of motion, we see in particular that the elements of $\Lambda_R^* C^\infty(B)^G$ Poisson commute.\footnote{
Their property \eqref{S6} implies
that the dressing invariant functions, $C^\infty(B)^G$,
form the center of the Poisson algebra $(C^\infty(B), \{- ,- \}_B)$ \eqref{PBBT}.}
The number of independent constants of motion is $2 \dim(B) - \ell$. This is a consequence  of
the identity
\be
 b_L^{-1} (b_L^{-1})^\tau = g_R^{-1} (b_R b_R^\tau) g_R
\label{S7}\ee
that leads to $\ell$ relations between the functions of $b_R$ and $b_L$.
We here used that $G$ acts by conjugations on the model $\fP$ of $B$ \eqref{nu}, and thus
\be
F( b_L^{-1} (b_L^{-1})^\tau) = F( b_R b_R^\tau), \quad \forall F \in C^\infty(\fP)^G.
\label{S8}\ee
The ring   $C^\infty(\fP)^G\simeq C^\infty(B)^G$ is  generated by $\ell=\rank(\cG)$ basic invariants,
which equals the functional dimension of $\Lambda_R^* C^\infty(B)^G$ \eqref{MHams} as well.
In conclusion, these Hamiltonians form a degenerate integrable system on $M$. Of course,
the same is true for the equivalent Hamiltonians $\pi_2^* C^\infty(B)^G$ \eqref{fMHams} on $\fM$.

\begin{prop}\label{prop:sol2}
Consider the Hamiltonian
$H = \Xi_R^* h$ with $h \in C^\infty(G)^G$.
Then, the corresponding evolution equation reads
\be
\dot{K} = K Dh(g_R).
\label{s1}\ee
The constituents in the decompositions $K= b_L g_R^{-1}= g_L b_R^{-1}$ \eqref{KdecT} satisfy
\be
\dot{b}_R = - Dh(g_R) b_R, \quad \dot{b}_L = b_L Dh(g_R), \quad \dot{g}_L = 0,\quad
\dot{g}_R = [(\ri \nabla h(g_R))_\cG, g_R].
\label{s2}\ee
The solution can be written as
\be
b_R(t) = \beta(t)^{-1} b_R(0), \quad b_L(t) =b_L(0) \beta(t), \quad g_L(t) = g_L(0), \quad g_R(t) = \gamma(t) g_R(0) \gamma(t)^{-1},
\label{s3}\ee
where $\beta(t)$ and $\gamma(t)$ are determined by the following factorization problem in $G^\bC_\bR$:
\be
\exp(\ri t \nabla h(g_R(0))) = \beta(t) \gamma(t) \quad\hbox{with}\quad
\beta(t) \in B,\, \gamma(t) = G.
\label{s4}\ee
Equivalently to \eqref{s3}, we have $K(t) = K(0) \beta(t)$.
\end{prop}
 \begin{proof}
 Lemma \ref{lm:B2} now gives
 \be
\nabla H(K) = - b_L D' h(g_R) b_L^{-1},
\qquad
\nabla' H(K) = - g_R D' h(g_R) g_R^{-1}.
\ee
Any function $h\in C^\infty(G)$ satisfies  $D h (g) = (g D'h(g) g^{-1})_\cB$, and  $Dh(g) = D'h(g)$ holds for $h\in C^\infty(G)^G$.
Thus we get
\be
\rho(\nabla' H(K))= ( g_R D'h(g_R) g_R^{-1})_\cB - \frac{1}{2}  g_R D' h(g_R) g_R^{-1} = Dh(g_R) - \frac{1}{2} g_R D h(g_R) g_R^{-1},
\ee
and $\rho(\nabla H(K)) = \frac{1}{2} b_L (D h(g_R)) b_L^{-1}$.
Inserting these into \eqref{HV} leads to \eqref{s1}:
\be
\dot{K}  = K(Dh(g_R) - \frac{1}{2} g_R D h(g_R) g_R^{-1}) + \frac{1}{2}  b_L (D h(g_R)) b_L^{-1} K = K Dh(g_R),
\ee
where the last equality relies on writing  $K= b_L g_R^{-1}$.
Since $Dh(g_R) \in \cB$, \eqref{s1} implies that $\dot{b}_R = - Dh(g_R) b_R$ and $\dot{g}_L = 0$. Then
$\dot{b}_L = b_L Dh(g_R)$ follows, because $\Lambda$ \eqref{mom} Poisson commutes with any $H\in C^\infty(M)^G$.
Due to \eqref{Dfromnab},
\be
 \ri \nabla h(g_R) = Dh(g_R) + (\ri \nabla h(g_R))_\cG,
 \label{s5}\ee
and $[g_R, \ri \nabla h(g_R)] =0$ because of the invariance property of $h$.
By using these relations,  the formula for $\dot{g}_R$ is derived from
\be
g_R = b_R g_L^{-1} b_L.
\label{s6}\ee
Turning to the solution, we first note that the curve $(g_R(t), b_R(t))$ defined by \eqref{s3}
 satisfies the differential equations
\be
\begin{aligned}
& \dot{g}_R(t) = [ \dot{\gamma}(t) \gamma(t)^{-1}, g_R(t)],\\
&\dot{b}_R(t) = - \beta(t)^{-1}\dot{\beta}(t) b_R(t).
\end{aligned}
\label{s7}\ee
Moreover, the equality \eqref{s4} implies
\be
\beta(t) \gamma(t) \ri \nabla h(g_R(0)) = \dot{\beta}(t) \gamma(t) + \beta(t) \dot{\gamma}(t).
\label{s8}\ee
From here, we get
\be
\ri \nabla h(g_R(t)) = \ri \gamma(t) \nabla h(g_R(0)) \gamma(t)^{-1} = \beta(t)^{-1} \dot{\beta}(t) + \dot{\gamma}(t) \gamma(t)^{-1},
\label{s9}\ee
where first equality holds because of the $G$-invariance of $h$.
We see from \eqref{s9} that
\be
\beta(t)^{-1} \dot{\beta}(t) = (\ri \nabla h(g_R(t)))_\cB
\quad\hbox{and}\quad
\dot{\gamma}(t) \gamma(t)^{-1} = (\ri \nabla h(g_R(t)))_\cG .
\label{s10}\ee
Inserting these relations into \eqref{s7}, we obtain
\be
\begin{aligned}
&\dot{b}_R(t) =(-\ri \nabla h(g_R(t)))_\cB\, b_R(t),\\
& \dot{g}_R(t) = [  (\ri \nabla h(g_R(t)))_\cG, g_R(t)].
\end{aligned}
\label{s11}\ee
Since $(\ri \nabla h(g_R))_\cB = Dh(g_R)$, comparison with \eqref{s2} shows that
the proof is complete.
\end{proof}

It is clear that $\Xi_R^* C^\infty(G)^G$ is generated by $\ell = \rank(\cG)$ functionally
independent Hamiltonians, which are in involution, since they remain constant along the
flows \eqref{s3}. To show their degenerate integrability, we observe that any
smooth real function of
\be
W(K):= b_L b_R g_L^{-1} = b_L g_R b_L^{-1}
\label{W}\ee
is a constant of motion.
Indeed, $\Lambda(K)= b_L b_R$ and $\Xi_L(K) = g_L$ are both constants of motion by \eqref{s2}.
We see from \eqref{W} that the set of the elements $W(K)$ is the union of those conjugacy classes in $G^\bC$ that
have representatives in $G_0<G$. Generically, the elements of this set can be parametrized by $(N-\ell)$ real variables,
where $N= 2 \dim(\cG)$ is the dimension of the Heisenberg double.
This holds since the generic elements of $G_0< G_\bR^\bC$ are fixed precisely by $G_0^\bC$ with respect to conjugations.
It follows that the functional dimension of the ring of joint constants of motion
of the Hamiltonians belonging to $\Xi_R^* C^\infty(G)^G$ is $(N-\ell)$,
and thus these Hamiltonians form a  degenerate integrable system.

\begin{rem}
It is a simple exercise to re-derive
the evolution equations for the variables $(g_R, b_R) = m(K)$ working directly with $(\fM,\{- ,- \})$ \eqref{A2T}.
In the above we have chosen to  use the model $(M,\{-,- \}_+)$ of the Heisenberg double since we wished
to present the time development of all constituents that enter $K= g_L b_R^{-1} = b_L g_R^{-1}$.
However, the model $(\fM, \{-,- \})$ will prove more convenient in what follows.
\end{rem}

\subsection{Deformation of spin Sutherland models from Poisson reduction}

Now we consider Poisson reduction based on  the $G$-action $\cA$ \eqref{cA} on $\fM$ \eqref{fM}.
This means that we keep only the $G$-invariant functions, and characterize their
Poisson brackets by restriction
to a convenient gauge slice.

We denote by $C^\infty(\fM)^G$ the ring of invariant functions. Any $\cF\in C^\infty(\fM)^G$ satisfies the identity
\be
D_1\cF(g,b) - D_1' \cF(g,b) + (b D_2' \cF(g,b) b^{-1})_\cB =0, \quad \forall (g,b)\in \fM,
\label{crucid2}\ee
as follows by  the taking derivative of $\cF \circ \cA_{e^{tX}}=\cF$ with respect to $t$, for every $X\in \cG$.
Here we utilized the decomposition \eqref{cGdec}, but below we shall also use the alternative decomposition
\be
\cG^\bC = \cG^\bC_0 + \cG^\bC_\perp, \qquad \cG^\bC_\perp \equiv \cG^\bC_< + \cG^\bC_>,
\ee
whereby we may write
\be
X= X_0 + X_\perp, \quad \forall X\in \cG^\bC.
\ee
Consider the connected subgroup $G_0^\bC < G^\bC$  corresponding to $\cG_0^\bC$.  By definition, the subset $G_{0,\reg}^\bC \subset G_0^\bC$
consists of those $g_0\in G_0^\bC$
for which $\Ad_{g_0} \in \End(\cG^\bC)$ is invertible on $\cG^\bC_\perp$. It is clear that $G_0^\reg \subset G_{0,\reg}^\bC$ and
for any $g_0 \in G_{0,\reg}^\bC$ we extend the definition \eqref{RQ1} by putting
\be
 \cR(g_0) (X) = \frac{1}{2} (\Ad_{g_0} + \id)\circ (\Ad_{g_0} - \id)_{\vert \cG_\perp^\bC}^{-1}(X_\perp),
 \qquad \forall g_0\in G^\bC_{0,\reg},\,\, X\in \cG^\bC.
 \label{RQ2}\ee

We introduce the following subsets of $\fM$ \eqref{fM},
\be
\fM^\reg:= \{ (g,b) \in \fM \mid g\in G^\reg \}, \quad \fM_0^\reg := \{ (Q,b) \in \fM \mid Q\in G_0^\reg \},
\ee
and observe that the restriction of functions provides an isomorphism
\be
C^\infty(\fM^\reg)^G \Longleftrightarrow  C^\infty(\fM_0^\reg)^\fN,
\label{isom1+}\ee
 using the normalizer $\fN$ \eqref{fN}.
For any $F\in C^\infty(\fM_0^\reg)$, we introduce the derivative $D_1 F(Q,b)\in \cB_0$ by
\be
\langle D_1 F(Q,b), X_0 \rangle_\bI = \dt F(e^{t X_0} Q, b), \qquad \forall X_0\in \cG_0.
\ee
The $\cG$-valued derivatives
$D_2 F$ and $D'_2 F$ are determined analogously to \eqref{derB}.

\begin{thm}\label{thm:defSuth}
 Let $F,H \in C^\infty(\fM_0^\reg)^\fN$ be the restrictions of $\cF, \cH \in C^\infty(\fM^\reg)^G$, and define
\be
\{ F,H\}_\red(Q,b) := \{\cF, \cH\}(Q,b).
\ee
Then equation \eqref{A2T} leads to the following formula of the reduced Poisson bracket:
\be
 \begin{aligned}
\{F,H\}_\red(Q,b)& = \langle D_1 F, D_2 H \rangle_\bI - \langle D_1 H, D_2 F \rangle_\bI \\
&+  \langle \cR(Q)(b D_2'H b^{-1})_\cB, D_2 F\rangle_\bI -  \langle \cR(Q)(b D_2'F b^{-1})_\cB, D_2 H\rangle_\bI,
 \end{aligned}
\label{RED1}\ee
where the subscript $\cB$ refers to \eqref{cGdec}, the derivatives are taken at $(Q,b)$, and
$\cR(Q)$ is given by \eqref{RQ2}.
\end{thm}
\begin{proof}
At any $(Q,b)\in\fM_0^\reg$, we have
\be
D_2' \cF(Q,b) = D_2' F(Q,b), \quad (D_1\cF(Q,b))_0 = (D_1' \cF(Q,b))_0 = D_1 F(Q,b),
\ee
where we use the decomposition \eqref{triang}.
Then the identity \eqref{crucid2} implies that
\be
( b D_2' \cF b^{-1})_{\cB_0} = 0,
\quad\hbox{with}\quad \cB= \cB_0 + \cB_>,
\label{invc}\ee
and
\be
( \Ad_{Q^{-1}} - \id) D_1 \cF(Q,b)_> = (b D_2' F(Q,b) b^{-1})_\cB.
\ee
This is solved by
\be
(D_1 \cF(Q,b))_> = -( \frac{1}{2} \id + \cR(Q)) (b D_2' F(Q,b) b^{-1})_\cB,
\ee
where we use the triangular decomposition \eqref{triang}.

We have to substitute the above relations into
\be
\{\cF, \cH\}(Q,b) =\left\langle D_2' \cF, b^{-1} (D_2\cH) b \right\rangle_\bI
 +  \left\langle D_1\cF , D_2\cH \right\rangle_\bI
-\left\langle D_1 \cH , D_2\cF \right\rangle_\bI,
\label{A2T+}\ee
which is obtained from \eqref{A2T} by noting that $\left\langle D'_1\cF, Q^{-1} (D_1\cH) Q\right\rangle_\bI=0$
since $\Ad_{Q^{-1}}$ maps $\cB$ to $\cB$.
At $(Q,b)$, because the Poisson bracket is antisymmetric,
\be
\left\langle D_2' \cF,  b^{-1} (D_2\cH) b \right\rangle_\bI = \frac{1}{2} \left\langle (b D_2'F  b^{-1})_\cB, D_2 H \right\rangle_\bI
- \frac{1}{2} \left\langle (b D_2'H  b^{-1})_\cB, D_2 F \right\rangle_\bI.
\ee
On the other hand, we get
\be
\left\langle D_1\cF , D_2\cH \right\rangle_\bI =\langle D_1 F  - (\frac{1}{2}\id + \cR(Q))  (b D_2' F(Q,b) b^{-1})_\cB, D_2 H\rangle_\bI,
\ee
and similar for $\left\langle D_1\cH , D_2\cF \right\rangle_\bI$.
The sum of these expressions gives us the formula \eqref{RED1}.
\end{proof}

\begin{rem}
Analogous to the reduced Poisson brackets \eqref{redcot1T} presented in Section \ref{S:Sec2}, formula \eqref{RED1} defines
a Poisson algebra structure not only on $C^\infty(\fM_0^\reg)^\fN$, but on the larger ring
$C^\infty(\fM_0^\reg)^{G_0}$ as well.
\end{rem}

Now we deal with the reduction of the dynamics induced by the Hamiltonians in $\pi_2^*(C^\infty(B)^G)$.
The projection of the Hamiltonian vector field of $\cH= \pi_2^* \phi$  on the quotient space $\fM^\reg/G$ descends also
from the evolution vector field living on $\fM_0^\reg$ that is described below.
This represents an intermediate step between the dynamics on $\fM^\reg$ and on $\fM^\reg/G$.

\begin{prop}\label{prop:REDeq1}
If $\cH = \pi_2^* \phi$ with $\phi\in C^\infty(B)^G$, i.e., $\cH(g,b) = \phi(b)$, then the formula
\eqref{RED1} of the reduced Poisson bracket simplifies to
\be
\{ F, H\}_\red(Q,b)= \langle D_1 F(Q,b), D \phi(b) \rangle_\bI
+ \langle D_2'F(Q,b), \left( b^{-1} (\cR(Q) D\phi(b)) b\right)_\cB \rangle_\bI,
\label{RED1+}\ee
and the corresponding reduced evolution equation on $\fM_0^\reg$ can be taken to be
\be
\dot{Q} =  (D \phi(b))_0 Q, \qquad
\dot{b} = b \left( b^{-1} (\cR(Q) D \phi(b)) b\right)_\cB.
\label{REDeq1}\ee
\end{prop}
\begin{proof}
For the restriction $H$ of $\cH$ on $\fM_0^\reg$, we get
\be
( b D_2' H(Q,b) b^{-1})_\cB = ( b D'\phi(b) b^{-1})_\cB = 0,
\ee
because $\phi \in C^\infty(B)^G$, as was noted before \eqref{S6}.
Moreover, we have $D_1 H(Q,b) = 0$ and, due to the antisymmetry of $\cR(Q)$,
\be
\langle \cR(Q)(b D_2'F b^{-1})_\cB, D_2 H\rangle_\bI = -
\langle D_2'F(Q,b), \left( b^{-1} (\cR(Q) D\phi(b)) b\right)_\cB \rangle_\bI.
\ee
Thus we obtain \eqref{RED1+} from \eqref{RED1}.
Then we see that the derivative of $F$ given by $\{F,H\}_\red$ coincides
with the derivative along the integral curves of the evolution equation \eqref{REDeq1},
which is the very justification of this equation.
Of course, the evolution equation on $\fM_0^\reg$ can be changed by adding any vector field that
vanishes upon projection on $\fM^\reg/G \simeq \fM^\reg_0/\fN$.
\end{proof}

Let us recall that the manifold $\fP$ \eqref{fP} can serve as a model of $B$ by means
of the diffeomorphism $\nu$ \eqref{nu}. For any function $\phi$ on $B$ we introduce the corresponding
function $\tilde \phi$ on $\fP$ by the definition
\be
\tilde \phi(L) = \phi(b)\quad \quad \hbox{with} \quad L :=\nu(b) = b b^\tau.
\label{Ldef}\ee
Then we define the derivative $\cD\tilde \phi(L)\in \cG$ by
\be
\langle \cD\tilde \phi(L), X\rangle_\bI = \dt \tilde \phi(e^{tX} L e^{t X^\tau}), \qquad \forall X\in \cB,
\ee
where used the notation \eqref{taumap}.
This implies the equality
\be
\cD\tilde \phi(L) = D\phi(b)
\quad\hbox{for}\quad L= b b^\tau.
\ee

\begin{prop}\label{prop:REDeq1+}
In terms of the variables $(Q,L)\in G_0^\reg \times \fP$ introduced in \eqref{Ldef}, the reduced evolution equation
\eqref{REDeq1} takes the form
\be
\dot Q = (\cD \tilde \phi(L))_0 Q, \qquad
\dot{L} = [ \cR(Q)(\cD \tilde \phi(L)), L],
\label{REDeq1+}\ee
which generalize the spin Sutherland evolution equation \eqref{redeq1}.
\end{prop}
\begin{proof}
Let us put $X:= \cR(Q)(\cD \tilde \phi(L))$, which belongs to $\cG$, and note that
\be
\left( b^{-1} (\cR(Q) D \phi(b)) b\right)_\cB = b^{-1} X b - (b^{-1} X b)_\cG.
\ee
Then, starting from \eqref{REDeq1}, we get
\be
\dot{L} = \dot{b} b^\tau + b {\dot{b}}^\tau =
Xb b^\tau - b (b^{-1} X b)_\cG b^\tau + b (b^\tau X^\tau - ((b^{-1} X b)_\cG)^\tau b^\tau) = [X, L],
\ee
because $Y^\tau = -Y$ for all $Y\in \cG$.
\end{proof}

It is an interesting exercise to recast the Poisson bracket \eqref{A2T} and its reduced version \eqref{A2T+} in terms
of the models $G\times \fP$ and $G_0^\reg \times \fP$ of $\fM$ and $\fM_0^\reg$.

\subsubsection{Reduced integrability and interpretation as deformed spin Sutherland models.}

Let us define the smooth,  $G$-equivariant map $\Psi_3: \fM \to \fP \times \fP$ by
\be
\Psi_3(g,b) := (\tilde L, L) \quad\hbox{with}\quad L= bb^\dagger \quad\hbox{and}\quad \tilde L := g^{-1} L g,
\label{Psi3}\ee
where $\eta \in G$ acts on $\fP \times \fP$ by conjugating  both components of $(a,b)\in \fP\times \fP$.
Then, for any  function $\chi\in C^\infty(\fP \times \fP)^G$, $\Psi_3^* (\chi)\in C^\infty(\fM)^G$
gives a smooth, $G$-invariant constant of motion.
Recalling that $\fP = \exp(\ri \cG)$, we see the close analogy with the constants of motion
observed in the cotangent bundle case (cf.~equation \eqref{Psi1}).
Thus, degenerate integrability on generic symplectic leaves of $\fM/G$ should hold in our present case as well.
We do not repeat the arguments of Section \ref{S:Sec2.1}, only make two remarks.
First,  $\ell=\rank(\cG)$ Casimir functions on $\fM/G$ arise
from the functions of the form $\phi \circ \Lambda \circ m^{-1}$, where $m: M\to \fM$ is given in \eqref{mMfM},
$\Lambda: M\to B$ is the moment map  \eqref{mom}, and $\phi\in C^\infty(B)^G$.
Second,
  $\pi_2^*(C^\infty(B)^G)$ gives rise to $\ell$ generically independent
Hamiltonians on $\fM/G$. These Hamiltonians are the $G$-invariant functions of $L$,
and they generically remain independent after fixing the just mentioned Casimir functions.

Let us consider a dressing orbit $\cO_B$ of $G$ in $B$, i.e., a symplectic leaf in the Poisson--Lie group $B$.
It is known from general theory \cite{KS,Lu,STS+} that the quotient spaces $\Lambda^{-1}(\cO_B)/G$ are Poisson subspaces of $M/G\simeq \fM/G$.
(They are stratified symplectic spaces  in general, which are unions of a dense open symplectic leaf and lower-dimensional strata \cite{OR,SL,Sn}.)
By using the Poisson--Lie version of symplectic reduction, we developed a detailed description of these
subspaces in \cite{F1}. Now, we translate the result into our present Poisson reduction setting.

Let $G_0^o$ be an arbitrarily chosen connected\footnote{
The  closure of $G_0^o$ in $G_0$ is homeomorphic to the space of conjugacy classes of $G$,
and is also homeomorphic to a convex polytope in $\cG_0$, a so-called Weyl alcove \cite{DK}.}
 component of $G_0^\reg$. Denote by $B_>$ the maximal nilpotent subgroup of $B$
associated with the subalgebra $\cB_> = \cB \cap \cG_>^\bC$ \eqref{cBtriang}, and consider also
$B_0:= \{ e^p \mid p \in \ri \cG_0\}$.
For arbitrarily given $Q\in G_0^o$ and $S_+ \in B_>$, the equation
\be
Q^{-1} b_+^{-1} Q b_+ S_+ = \1_B
\ee
admits a unique solution for $b_+\in B_>$, which defines the function $b_+(Q,S_+)$.
Then, all the elements $(Q,b) \in \fM_0^\reg$ with $Q\in G_0^o$ can be uniquely written in the form
\be
b= e^{p} b_+(Q,S_+)
\quad \hbox{with}\quad (Q,p, S_+) \in G_0^o \times \ri \cG_0 \times B_>,
\ee
and this induces the identification
\be
\fM^\red/G \simeq \fM_0^\reg/\fN \simeq  G_0^o \times \ri \cG_0 \times (B_>/G_0).
\ee
In this parametrization of the Poisson quotient the components of $q$ in $Q=e^{\ri q}$ and $p$ form canonically conjugate pairs,
and $S_+ \in B_>$ is a `collective spin degree of freedom' that decouples from $q$ and $p$ under the
reduced Poisson bracket.
The space $B_>/G_0$ represents the reduction $B$ with respect to $G_0<G$, at the zero value
of the classical moment map that generates the conjugation action of $G_0$ on $B$.
The `main reduced Hamiltonians' are obtained by taking the trace of
\be
L(Q,p,S_+) = e^p b_+(Q,S_+) b_+(Q,S_+)^\tau e^p
\label{defLax}\ee
 in the fundamental irreducible representations of $G^\bC$.
 In \cite{F1}, the structure of
$b_+(Q,S_+)$ was  elaborated (for $G=\mathrm{U}(n)$  even its fully
explicit formula was given),
and by using this it was shown that the Lax matrices $L(Q,p,S_+)$  and the main Hamiltonians
 of the models at issue are deformations of the Lax matrices
\eqref{Jpar}
and main Hamiltonians
of the spin Sutherland models \eqref{spinSuthHam}.
For the details of these results, one can consult \cite{F1}.

\subsection{The duals of the deformed spin Sutherland models}

Now, we describe the reduction of the integrable system of
Proposition \ref{prop:sol2} by utilizing that any element of $B$ can be
transformed into the subgroup $B_0 = B \cap G^\bC_0$ by the dressing action of $G$.

Let us introduce
\be
B_0^\reg:= B_0 \cap G^\bC_{0,\reg}
\ee
and denote by $B^\reg$ the union of the $G$-orbits in $B$ that intersect $B_0^\reg$.
Then define
\be
\FMpreg:= \{ (g,b) \in \fM \mid  b\in B^\reg \}, \quad \fMpreg0 := \{ (g,\Gamma) \in \fM \mid \Gamma\in B_0^\reg \}.
\ee
We remark that all powers of $\Gamma\in B_0^\reg$ belong to $B_0^\reg$.
Similarly to \eqref{isom1+}, the restriction of functions provides the isomorphism
\be
C^\infty(\fM'^\reg)^G \Longleftrightarrow  C^\infty(\fMpreg0)^\fN.
\label{isom2+}\ee
Our aim is to find the formula for the Poisson bracket on $C^\infty(\fMpreg0)^\fN$ induced by
this isomorphism.
The derivation follows the steps of the previous section, but it is slightly more complicated.
In preparation, we introduce  $\varrho(\Gamma) \in \End(\cG^\bC)$ by
\be
\varrho(\Gamma)(X) = (\sinh \ad_\gamma)^{-1}(X_\perp),\qquad \forall \Gamma = e^\gamma \in B_0^\reg,\,\, X=(X_0+X_\perp)\in \cG^\bC.
\label{vrho}\ee
This is well-defined\footnote{In fact, $\cG^\bC_\perp$ is spanned by
the root vectors $E_\alpha$ and $\alpha(\gamma)$ is a nonzero real number for $e^\gamma \in B_0^\reg$.
See Appendix \ref{sec:A}.} due to the definition of $B_0^\reg$.
Note that $\varrho(\Gamma)$ vanishes on $\cG^\bC_0$ and it maps $\cB_>$ to itself.
Below we shall apply the operator $\cR(\Gamma^2)$ \eqref{RQ2}, which can also be written as
\be
\cR(\Gamma^2)(X) = \frac{1}{2} (\coth \ad_\gamma) (X_\perp).
\label{RGam}\ee
For any $F\in C^\infty(\fMpreg0)$, the derivative $D_2 F(g,\Gamma)\in \cG_0$ is determined by
\be
\langle D_2 F(g,\Gamma), X_0 \rangle_\bI = \dt F(g, e^{t X_0}  \Gamma), \qquad \forall X_0\in \cB_0,
\ee
and the $\cB$-valued derivatives
$D_1 F$ and $D'_1 F$ are determined  analogously to \eqref{derG}.

Let $F\in C^\infty(\fMpreg0)^\fN$ be the restriction of $\cF \in C^\infty(\FMpreg)^G$.
Then we obviously have
\be
D_1 \cF(g,\Gamma) = D_1 F(g,\Gamma),\,\,
D_1' \cF(g,\Gamma) = D_1' F(g,\Gamma),
\qquad
 D_2\cF(g,\Gamma)_0 = D_2' \cF(g,\Gamma)_0 = D_2 F(g,\Gamma),
\label{cFandF1}\ee
and the residual $G_0$-invariance of $F$ implies
\be
D_1 F(g,\Gamma)_0 = D_1'F(g,\Gamma)_0.
\ee
The full expression of $D_2\cF(g,\Gamma)$ through the derivatives of $F$ is given by the next lemma.

\begin{lem}  \label{LemDers}
Let $F\in C^\infty(\fMpreg0)^\fN$ be the restriction of $\cF \in C^\infty(\FMpreg)^G$.
At any fixed $(g,\Gamma)\in \fMpreg0$,
 put
\be
X_0:= D_2 F\in \cG_0, \quad Y:= (D_1' F - D_1 F)\in \cB_> .
\ee
Then, the invariance condition \eqref{crucid2} reads
\be
(\Gamma D_2' \cF \Gamma^{-1})_\cB = Y,
\label{tosolve}  \ee
and this implies  the formula
 \be
 D_2' \cF = X_0 + \frac{1}{2} \varrho(\Gamma) ( Y + Y^\tau).
 \ee
Furthermore, we have
\be
\Gamma \cD_2'\cF \Gamma^{-1} =  X_0 + \frac{1}{2} (Y + Y^\tau) +  \cR(\Gamma^2) ( Y + Y^\tau),
\ee
and
\be
D_2 \cF  \equiv (\Gamma D_2' \cF \Gamma^{-1})_\cG =  X_0 +  \frac{1}{2} ( Y^\tau - Y) +  \cR(\Gamma^2) ( Y + Y^\tau).
\ee
Here, $\Gamma = e^\gamma\in B_0^\reg$ and the operators  \eqref{vrho},  \eqref{RGam} are employed.
For the definition of $Y^\tau$, see equation \eqref{taumap} and Appendix \ref{sec:A}.
\end{lem}

\begin{proof}
Let us note  that any $V\in \cG^\bC$ can be decomposed as
\be
V = V_> + V_0 + V_< = V_> + V_0^\ri + V_0^\r + V_<, \quad V_0^\ri\in \cG_0,\, V_0^\r\in \cB_0,
\ee
and then
\be
V_\cB = V_> + (V_<)^\tau + V_0^\r,\qquad
V_\cG = V_0^\ri + V_< - (V_<)^\tau.
\label{Vdec}\ee
According to \eqref{tosolve}, we need to solve an equation of the form
\be
(\Gamma X \Gamma^{-1})_\cB = Y
\ee
for $X$, where $Y = Y_> \in \cB_>$ and $X =  (X_> + X_< + X_0^\ri) \in \cG$.
By using that $\Gamma^\tau = \Gamma$  and $X_> = - (X_<)^\tau$, we get
\be
(\Gamma X \Gamma^{-1})_\cB = X_0^\ri + \Gamma X_> \Gamma^{-1} - \Gamma^{-1} X_> \Gamma = X_0^\ri + 2 (\sinh \ad_\gamma) (X_>).
\ee
From here, we get
\be
X_> = \frac{1}{2} \varrho(\Gamma) (Y).
\ee
Since $X_< = - (X_>)^\tau$ and $(\ad_\gamma V)^\tau = - \ad_\gamma V^\tau$ for all $V\in \cG^\bC$, this implies
\be
X= X_0^\ri + \frac{1}{2} \varrho(\Gamma) (Y + Y^\tau).
\ee
Then we find
\be
\Gamma X \Gamma^{-1} =(\sinh\ad_\gamma + \cosh \ad_\gamma)(X)
= X_0^\ri + \frac{1}{2} (Y + Y^\tau) + \cR(\Gamma^2)( Y + Y^\tau).
\ee
Notice that $\cR(\Gamma^2) (Y + Y^\tau) \in \cG$.
Consequently, by applying \eqref{Vdec} to $(Y+ Y^\tau)$ with $Y = Y_>$, we obtain
\be
(\Gamma X \Gamma^{-1})_\cG = X_0^\ri + \frac{1}{2}(Y^\tau - Y) + \cR(\Gamma^2)(Y+ Y^\tau).
\ee
By taking $X= \cD_2' \cF(g,\Gamma)$, the proof is finished.
 \end{proof}

Let us recall that the functions $f$ on $G$ have  the $\cG$-valued
derivatives $\nabla f$ and $\nabla' f$ defined in \eqref{nabG},
and (as seen from \eqref{Dnab}) these are related to the $\cB$-valued derivatives $Df$ and $D'f$ by
\be
\ri \nabla f(g) = \frac{1}{2} ( D f(g) + (D f(g))^\tau),\,\,\,
\ri \nabla' f(g)  = \frac{1}{2} ( D' f(g) + (D' f(g))^\tau).
\label{gnab}\ee
In consequence  of \eqref{Ri} and \eqref{Dnab}, one also has
\be
R^\ri \nabla f(g)= (-\ri \nabla f(g))_\cG = \frac{1}{2} (  D f(g) - (D f(g))^\tau ),\,\,\,
R^\ri \nabla' f(g) = (- \ri \nabla' f(g))_\cG = \frac{1}{2} ( D' f(g) -(D' f(g))^\tau ).
\label{gnab+}\ee
These relations are applied below to functions on $\fMpreg0=G \times B_0^\reg$, regarding
the derivatives with respect to the first variable.

\begin{thm}\label{thm:RED2}  Let $F,H\in C^\infty(\fMpreg0)^\fN$ be the restrictions of the functions
$\cF, \cH \in C^\infty(\FMpreg)^G$,
respectively, and define
\be
\{ F,H\}_\red'(g,\Gamma) := \{ \cF, \cH\}(g,\Gamma), \qquad \forall (g,\Gamma)\in \fMpreg0.
\ee
Then, equation  \eqref{A2T} implies the following formula for this reduced Poisson bracket:
\be
 \begin{aligned}
\{F,H\}'_\red(g,\Gamma) & = \langle \nabla_1 F, D_2 H\rangle_\cG - \langle \nabla_1 H , D_2 F\rangle_\cG \\
&+  2\langle \nabla_1' F , \cR(\Gamma^2) (\ri \nabla_1' H) \rangle_\cG - 2\langle \nabla_1 F, \cR(\Gamma^2) (\ri \nabla_1 H)\rangle_\cG,
 \end{aligned}
\label{RED2}\ee
where $\nabla_1 F$ and $\nabla_1 H$ denote the $\cG$-valued derivatives defined similarly to \eqref{nabG}, and $\cR(\Gamma^2)$ \eqref{RGam} is used.
\end{thm}
\begin{proof}
We have to evaluate the formula \eqref{A2T} at $(g,\Gamma)\in G \times B_0^\reg$
for invariant functions.
 By using the equalities \eqref{cFandF1} and Lemma \ref{LemDers}
we find
\be
\begin{aligned}
&\langle D_1\cF , D_2\cH \rangle_\bI
-\langle D_1 \cH , D_2\cF \rangle_\bI  =\langle D_1 F, D_2 H \rangle_\bI - \langle D_1 H, D_2 F \rangle_\bI \\
& \quad +  \langle D_1 F, (\cR(\Gamma^2) + \frac{1}{2}\id) (D_1' H - D_1 H)^\tau \rangle_\bI -
 \langle D_1 H, (\cR(\Gamma^2) + \frac{1}{2}\id)(D_1' F - D_1 F)^\tau \rangle_\bI.
\end{aligned}
\ee
We took into account that $\cR(\Gamma^2)$ maps $\cB$ to itself and that $\langle X,Y\rangle_\bI=0$
for any $X,Y\in \cB$.
Next, direct substitution gives
\be
\langle \Gamma D_2' \cF \Gamma^{-1}, D_2 \cH \rangle_\bI =  \langle D_1'F - D_1F , (\cR(\Gamma^2) + \frac{1}{2}\id) (D_1' H - D_1 H)^\tau\rangle_\bI.
\ee
We now collect terms, and in doing so employ the antisymmetry of $\cR(\Gamma^2)$ together with the properties
\be
\langle U^\tau, V^\tau \rangle_\bI = - \langle U, V \rangle_\bI, \quad \cR(\Gamma^2)(U^\tau) = -(\cR(\Gamma^2)U)^\tau,
\quad \forall U,V\in \cG^\bC,
\ee
which follow from the definitions. This gives
\be
\begin{aligned}
&\langle D_1\cF , D_2\cH \rangle_\bI
-\langle D_1 \cH , D_2\cF \rangle_\bI  + \langle \Gamma D_2' \cF \Gamma^{-1}, D_2 \cH \rangle_\bI =
\langle D_1 F, D_2 H \rangle_\bI - \langle D_1 H , D_2 F \rangle_\bI\\
& \quad + \langle D_1' F, \cR(\Gamma^2) (D_1'H)^\tau \rangle_\bI  -  \langle D_1 F, \cR(\Gamma^2) (D_1 H)^\tau \rangle_\bI
+ \frac{1}{2} \langle D_1' F, (D_1'H)^\tau \rangle_\bI   - \frac{1}{2} \langle D_1 F, (D_1 H)^\tau \rangle_\bI.
\end{aligned}
\label{2sum}\ee
Since $D_2 H\in \cG_0$, we have
\be
\langle D_1 F, D_2 H \rangle_\bI = \langle (D_1 F)_0, D_2 H\rangle_\bI
 = \langle \ri (\nabla_1 F)_0, D_2 H\rangle_\bI = \langle \nabla_1 F, D_2 H\rangle_\cG.
\label{zerograde}\ee
Referring \eqref{gnab+}, we can write
\be
\langle D_1' F, \cR(\Gamma^2) (D_1'H)^\tau \rangle_\bI = \langle D_1' F, \cR(\Gamma^2) ((D_1'H)^\tau + D_1'H)  \rangle_\bI
 = 2 \langle \nabla_1' F, \cR(\Gamma^2) (\ri \nabla_1' H)  \rangle_\cG,
\ee
where the last step holds since $\langle D_1' F, X\rangle_\bI = \langle \nabla_1' F, X\rangle_\cG$ for all $X\in \cG$,
and $\cR(\Gamma^2) (\ri X) \in \cG$ for all $X\in \cG$.
Therefore, the first 4 terms of \eqref{2sum} yield the right-hand side of \eqref{RED2}.
The rest of the terms cancel, because (by \eqref{gnab+}) we have
\be
\frac{1}{2} \langle D_1' F, (D_1'H)^\tau \rangle_\bI   - \frac{1}{2} \langle D_1 F, (D_1 H)^\tau \rangle_\bI
= \langle \nabla_1 F, R^\ri \nabla_1 H \rangle_\cG - \langle \nabla_1' F, R^\ri \nabla_1' H \rangle_\cG,
\label{PLGid}\ee
and  this is just the opposite of the remaining term  $-\langle(  D_1' F , g^{-1} (D_1 H) g \rangle_\bI $ of \eqref{A2T}.
This holds due to the identity \eqref{GPL}.
\end{proof}

\begin{rem}
The formula \eqref{RED2} defines a Poisson bracket not only on the ring $C^\infty(\fMpreg0)^\fN$, but on the
larger ring
$C^\infty(\fMpreg0)^{G_0}$, too.
The proof shows that one may rewrite it in the alternative form
\be
 \begin{aligned}
\{F,H\}'_\red(g,\Gamma) & = \langle \ri (\nabla_1 F)_0, D_2 H\rangle_\bI - \langle \ri (\nabla_1 H)_0, D_2 F\rangle_\bI\\
&+  \langle D_1' F , \cR(\Gamma^2) (D_1' H)^\tau )\rangle_\bI  - \langle D_1 F , \cR(\Gamma^2) (D_1 H)^\tau )\rangle_\bI.
 \end{aligned}
\label{RED2+}\ee
One can also use the identity \eqref{PLGid} to get an alternative formula the Poisson--Lie struture \eqref{GPL} on $G$.
It may be worth noting that
\be
\langle \nabla_1 F, D_2 H\rangle_\cG= \langle \ri \nabla_1 F, D_2 H\rangle_\bI =\langle \ri (\nabla_1 F)_0, D_2 H\rangle_\bI=
\langle (D_1 F)_0, D_2 H\rangle_\bI = \langle D_1 F, D_2 H\rangle_\bI,
\ee
\ since $D_2 H\in \cG_0$.
 \end{rem}

\begin{prop}\label{prop:REDeq2}
If $\cH = \pi_1^* h$ with $h\in C^\infty(G)^G$, i.e., $\cH(g,b) = h(g)$, then formula
\eqref{RED2} of the reduced Poisson bracket becomes
\be
\{F,H\}'_\red(g,\Gamma)  = - \langle  D_2 F(g,\Gamma), \ri (\nabla h(g))_0)\rangle_\bI
+  2\langle \nabla_1' F(g, \Gamma) - \nabla_1 F(g,\Gamma), \cR(\Gamma^2) (\ri \nabla h(g))\rangle_\cG.
\label{RED2++}\ee
Introducing  the new variable $P:= \Gamma^2$, the  reduced evolution equation induced by the
Hamiltonian $H$
on $\fMpreg0$ can be written as
\be
\dot{g} =  2 [ g, \cR(P) (\ri \nabla h(g))], \qquad \dot{P} = - 2 \ri (\nabla h(g))_0 P,
\label{REDeq2}\ee
which is quite analogous to \eqref{REDeq1+}.
\end{prop}

As a consistency check, we verified that the reduced evolution equation \eqref{REDeq2}  results also by applying the projection method to the corresponding
unreduced evolution  equation \eqref{s11}. Of course, the evolution equation on $\fMpreg0$ is unique only up
to infinitesimal gauge transformations that do not change its eventual projection on
${\fM'}^\reg/G \simeq \fMpreg0/\fN$.

In the $G=\mathrm{SU}(n)$ case the reduced system characterized by Theorem \ref{thm:RED2} and Proposition \ref{prop:REDeq2} gives \cite{FK1}
a special real form of the trigonometric Ruijsenaars--Schneider system on a small
symplectic leaf of dimension $2(n-1)$ in $\fMpreg0/\fN$. Thus one may expect to obtain spin
Ruijsenaars--Schneider type systems on generic symplectic leaves.
However,  it is not
known how to introduce positions, momenta and spin variables
in such a way that would endow the reductions of the pullback invariants $\pi_1^*(C^\infty(G)^G)$
with a many-body  interpretation. This is analogous to the open problem
that exists in relation to the second kind of reduced systems obtained from $T^*G$.

\subsubsection*{Toward integrability after reduction.}
We here explain that the mechanism described around equation \eqref{Psired} is applicable in the present case, too.
For this, we use the original model  $(M,\{- ,- \}_+)$ \eqref{A1T} of the Heisenberg double, and define the map
$\Psi_4: M\to G^\bC$ by
\be
\Psi_4(K) := W(K) \equiv  \Lambda_L(K) \Xi_R(K) \Lambda_L(K)^{-1}
\ee
using the formula \eqref{W} and the definitions \eqref{XiLaT}.
We have seen that this map is constant along the Hamiltonian flows generated by the
pullback invariants $\Xi_R^*( C^\infty(G)^G)$.
The conjugation action of $G$ on $G^\bC$ is defined by the maps $\cC_\eta$,
\be
\cC_\eta(K) := \eta K \eta^{-1}, \qquad \forall K\in G^\bC,\, \eta\in G.
\label{cCeta}\ee
We wish to show that $\Psi_4$ is equivariant with respect to the quasi-adjoint action \eqref{cA1}  and the conjugation action \eqref{cCeta},
\be
\Psi_4 \circ \cA^1_\eta = \cC_\eta \circ \Psi_4, \qquad \forall \eta\in G.
\label{Psi4equivar}\ee
In order to derive this, notice from \eqref{cA1} that
\be
\Lambda_L(\cA_\eta^1(K)) = \Lambda_L (\eta \Lambda_L(K))
\quad\hbox{and}\quad
\Xi_R (\cA_\eta^1(K))  = \Xi_R(\eta \Lambda_L(K))^{-1} \Xi_R(K)\Xi_R(\eta \Lambda_L(K)).
\ee
Therefore, we obtain
\be
\begin{aligned}
\Psi_4 (\cA^1_\eta(K))& =  \Lambda_L (\eta \Lambda_L(K)) \Xi_R(\eta \Lambda_L(K))^{-1}
\Xi_R(K)\Xi_R(\eta \Lambda_L(K))\Lambda_L (\eta \Lambda_L(K))^{-1}\\
& = \eta \Lambda_L(K) \Xi_R(K) (\eta \Lambda_L(K))^{-1} = \eta \Psi_4(K) \eta^{-1}.
\end{aligned}
\ee
Let us then consider the ring of $G$-invariant real functions
\be
C^\infty(G^\bC)^G_{\mathrm{conj}}:= \{ F \in C^\infty(G^\bC)\mid F \circ \cC_\eta = F,\quad \forall \eta\in G\}.
\ee
For every $F\in C^\infty(G^\bC)^G_{\mathrm{conj}}$, the function $F \circ \Psi_4\in C^\infty(M)$ is a joint constant of motion for
the pullback invariants
$\Xi_R^*(C^\infty(G)^G)$, and we see from \eqref{Psi4equivar} that  this function is invariant with respect to the
quasi-adjoint action of $G$ on $M$.
The alternative formula (cf. \eqref{W})
\be
\Psi_4(K) = \Lambda(K) \Xi_L(K)^{-1}
\ee
shows that the above constants of motion contain the elements of $\Lambda^*(C^\infty(B)^G)$, where $\Lambda$ is the moment map \eqref{mom}.
All  constants of motion $F\circ \Psi_4$ descend to the reduced space $M/G$, and those that depend only on $\Lambda$
become numerical constants on the symplectic leaves of this (stratified) Poisson space.
Thus, relying on a straightforward counting argument, we expect that the pullback invariants $\Xi_R^*(C^\infty(G)^G)$
engender degenerate integrable systems on generic symplectic leaves in $M/G$.

By using the diffeomorphism $m: M\to \fM$ \eqref{mMfM}, one can transfer the above construction to the alternative framework based
on the unreduced phase space $\fM = G \times B$.

Finally, let us note that the two kinds of reduced systems described in this section are subject
to a similar duality relation that we outlined at the end of Section \ref{S:Sec2}.

\begin{rem}
Let us consider the function $\chi_\rho(K):= \tr_\rho(K)$ on $G^\bC$, where $\rho$ is some finite dimensional
irreducible representation of $G^\bC$. Then we obtain $\tr_\rho(W(K))= \tr_\rho(\Xi_R(K))$.
This shows that the constants of motion $F\circ \Psi_4$ contain the basic pullback invariants associated with
the real and imaginary parts of the characters of the irreducible representations of $G$.
\end{rem}

\section{Reduction of the quasi-Poisson double $G\times G$}
\label{S:Sec4}

Quasi-Hamiltonian manifolds \cite{AMM} and quasi-Poisson manifolds  \cite{AKSM} were introduced  primarily in order to
provide a purely finite dimensional construction of the symplectic and Poisson structures
of  moduli spaces of flat connections, which were originally obtained by infinite dimensional
symplectic reduction. Since then, interesting applications of these concepts came to light in several fields,
including the construction of finite dimensional integrable Hamiltonian systems \cite{CF,FFM,FK2,FeKlu}.
The content of this section is closely related to the work Arthamonov and Reshetikhin \cite{AR}, who
constructed degenerate integrable systems on moduli spaces of flat connections, working mostly in a complex
holomorphic setting.

Let us recall that a quasi-Poisson manifold is a $G$-manifold, here denoted $\cS$, equipped with a $G$-invariant bivector, $\Pi$, whose
key property is that the formula
\be
\{\cF, \cH\} := (d\cF \otimes d\cH)(\Pi)
\label{qPB1}\ee
defines a Poisson algebra structure \emph{on the space of invariant function} $C^\infty(\cS)^G$.
By the standard identification $C^\infty(\cS/G) \equiv C^\infty(\cS)^G$, this leads to a Poisson structure on the quotient space.
Defining the
\emph{quasi-Poisson bracket} of any smooth functions by \eqref{qPB1}, we may associate a \emph{quasi-Hamiltonian vector field} $V_{\cH}$ to any
$\cH\in S^\infty(\cS)$ by putting
\be
 V_{\cH}[\cF]:= \{ \cF,\cH \} .
\label{qVH}\ee
The vector field $V_\cH$ descends to $\cS/G$ if $\cH$ is $G$-invariant. This means that the process of taking the quotient by the $G$-action
works for quasi-Poisson manifolds in the same way as it does for Poisson manifolds.
The quotient space is known to be a disjoint union of smooth symplectic manifolds, just as for reductions
defined by Hamiltonian actions of compact Lie groups \cite{OR,SL,Sn}.

For general functions, the quasi-Poisson bracket \eqref{qPB1} violates the Jacobi identity in a specific manner.
For this and further details,
one may consult \cite{AKSM}.

An important quasi-Poisson manifold is the double
\be
\fS:= G \times G
\label{D}\ee
of a connected compact Lie group.  For our purpose, we continue to assume that the Lie algebra $\cG$ of $G$ is simple and
$G$ is simply connected.
To describe the relevant bivector, we take a basis, $e_a$, of the Lie algebra $\cG$ that satisfies
\be
\langle e_a, e_b \rangle_\cG = - \delta_{a,b},
\ee
where, as before, we use the negative definite invariant bilinear form on  $\cG$.
To simplify the appearance of the formulas, below we omit the subscript $\cG$ from $\langle - ,- \rangle_\cG$.
We let $e_a^L$ and $e_a^R$ denote the left-invariant and right-invariant vector fields on $G$ that
extend the tangent vector $e_a$ at the unit element of $G$.
Furthermore, we let $e_a^{1,L}$ and $e_a^{2,L}$ stand for the corresponding vector fields on $\fS$ that are tangent to
the first and second $G$ factors of $\fS$, respectively. Similarly, we introduce $e_a^{1,R}$ and $e_a^{2,R}$ as well.
Then, the bivector $\Pi$ of the (internally fused) double \cite{AKSM} is given by\footnote{This is obtained from example 5.3 of
 \cite{AKSM} by performing the fusion procedure of Proposition 5.1 of this reference.}
 \be
2\Pi = e_a^{1,L}\wedge e_a^{2,R} + e_a^{1,R}\wedge e_a^{2,L} +e_a^{1,R}\wedge e_a^{1,L} +e_a^{2,L}\wedge e_a^{2,R}-
e_a^{1,R}\wedge e_a^{2,R} - e_a^{2,L}\wedge e_a^{1,L}.
\label{DP}\ee
Here and below, the summation convention is in force; the wedge product does not contain $\frac{1}{2}$.
This bivector is invariant with respect to the conjugation action of $G$ on $\fS$, where $\eta \in G$ operates by the map
\be
\cA_\eta: (g_1,g_2) \mapsto (\eta g_1 \eta^{-1}, \eta g_2 \eta^{-1}).
\ee
The $\cG$-valued derivatives of $\cF\in C^\infty(\fS)$, defined in the same way as in \eqref{nabG}, can be written as
\be
\nabla_1 \cF = - e_a e_a^{1,R}[\cF],
\qquad
\nabla_1' \cF =-e_a e_a^{1,L}[\cF],
\ee
and similarly for $\nabla_2 \cF$ and $\nabla_2' \cF$.   To be clear about the notations,  note that
\be
e_a^{1,R}[\cF](g_1,g_2) = \dt \cF( \exp(t e_a) g_1, g_2), \qquad \forall (g_1,g_2) \in \fS.
\ee
Then, the quasi-Poisson bracket \eqref{qPB1} of any two smooth functions $\cF$ and $\cH$ takes the following form:
\be
\begin{aligned}
2\{ \cF,\cH\} & = \langle \nabla_1' \cH, \nabla_2 \cF \rangle -\langle \nabla_2 \cH, \nabla_1' \cF \rangle+  \langle \nabla_1 \cH, \nabla_2' \cF \rangle
-\langle \nabla_2' \cH, \nabla_1 \cF \rangle \\
& +\langle \nabla_2 \cH, \nabla_1 \cF \rangle -\langle \nabla_1 \cH, \nabla_2 \cF \rangle+  \langle \nabla_1' \cH, \nabla_2' \cF \rangle
-\langle \nabla_2' \cH, \nabla_1' \cF \rangle \\
 &+ \langle \nabla_1 \cH, \nabla_1' \cF \rangle -\langle \nabla_1' \cH, \nabla_1 \cF \rangle+  \langle \nabla_2' \cH, \nabla_2 \cF \rangle
-\langle \nabla_2 \cH, \nabla_2' \cF \rangle.
\end{aligned}
\label{qSPB}\ee
This is obtained from \eqref{DP} by using identities like
\be
e_a^{1,L}[\cF] e_a^{2,R}[\cH] = - \langle \nabla'_1 \cF, \nabla_2 \cH \rangle.
\ee

In the next statement, we apply
 the natural projections $\pi_1$ and $\pi_2$ from $\fS$ to $G$,
\be
\pi_1(g_1,g_2):= g_1, \qquad \pi_2(g_1,g_2) := g_2.
\ee

\begin{prop}\label{prop:qsol}
Consider arbitrary functions $\cF \in C^\infty(\fS)$ and $\phi\in C^\infty(G)^G$.
Then, we have
\be
\{\cF, \pi_2^* \phi\}(g_1,g_2) = -\langle \nabla_1' \cF(g_1,g_2), \nabla\phi(g_2)\rangle,
\qquad
\{\cF, \pi_1^* \phi\}(g_1,g_2) = \langle \nabla_2' \cF(g_1,g_2), \nabla\phi(g_1)\rangle.
\label{qPB2}\ee
Thus, the quasi-Hamiltonian vector field \eqref{qVH} of $\pi_2^* \phi$ induces the evolution equation
\be
\dot{g}_1 =-g_1 \nabla\phi(g_2), \quad \dot{g}_2 = 0
\quad\hbox{whose solution is}\quad
 (g_1(t), g_2(t)) = (g_1(0) \exp( -t \nabla \phi(g_2(0))), g_2(0)),
 \label{qsolpi2}\ee
 and
 the quasi-Hamiltonian vector field of $\pi_1^* \phi$ induces the evolution equation
 \be
\dot{g}_1 =0, \quad \dot{g}_2 = g_2 \nabla \phi(g_1)
\quad\hbox{whose solution is}\quad
 (g_1(t), g_2(t)) = (g_1(0), g_2(0) \exp( t \nabla \phi(g_1(0))).
 \label{qsolpi1}\ee
 \end{prop}
 \begin{proof}
 Consider, for example, $\cH= \pi_1^* \phi$. Then, $\nabla_2 \cH = \nabla_2' \cH =0$ and $\nabla_1 \cH = \nabla_1' \cH = \pi_1^* \nabla \phi$.
In this case, simply by collecting terms,
\be
2 \{\cF,\cH\} = 2 \langle \nabla_2'\cF, \pi_1^* \nabla \phi \rangle
 +\langle \nabla_1' \cF - \nabla_1 \cF, \pi_1^* \nabla \phi \rangle,
\ee
and the second term vanishes on account of the relations
\be
\nabla_1 \cF(g_1, g_2) = g_1 \nabla_1' \cF(g_1,g_2) g_1^{-1},
\quad
g_1^{-1} \nabla\phi(g_1) g_1= \nabla \phi(g_1),
\ee
and the $G$-invariance of $\langle -,- \rangle$.
The rest of the statement is verified by straightforward inspection.
\end{proof}

We see from Proposition \ref{prop:qsol}
that the ring $\pi_2^* C^\infty(G)^G$  forms an Abelian Poisson algebra,
and $g_2$ as well as $\tilde g_1:= g_1 g_2 g_1^{-1}$
 are constant along all of the
corresponding integral curves \eqref{qsolpi2}.
This shows that the functional dimension
of
 the joint constants of motion
for the evolution equations in \eqref{qsolpi2} is $\dim(\fS) - \rank(\cG)$.
In conclusion, the family of Hamiltonians $\pi_2^* C^\infty(G)^G$, of functional dimension $\rank(\cG)$,
behaves basically in the same way as a degenerate integrable system on a symplectic manifold.
Quite similar observations apply to the Poisson algebra $\pi_1^* C^\infty(G)^G$.
We merely note that the relevant constants of motion are now provided by arbitrary smooth functions
of $g_1$ and  $\tilde g_2:= g_2 g_1 g_2^{-1}$.

Mimicking the reduction procedure of Section \ref{S:Sec2}, we introduce the submanifolds
\be
\fS^\reg:=\{ (g_1,g_2)\in \fS \mid g_1 \in G^\reg\}, \qquad
\fS_0^\reg :=\{ (Q,g) \in \fS \mid Q\in G_0^\reg\},
\label{S01}\ee
and
\be
\FSpreg:=\{ (g_1,g_2) \in \fS \mid g_2\in G^\reg\}, \qquad
\fSpreg0:=\{ (g,Q)\in \fS\mid  Q \in G_0^\reg\}.
\label{S02}\ee
Using the normalizer $\fN$ \eqref{fN}, restriction of functions engenders the
isomorphisms
\be
C^\infty(\fS^\reg)^G \Longleftrightarrow  C^\infty(\fS_0^\reg)^\fN
\label{Sisom1}\ee
and
\be
C^\infty(\FSpreg)^G \Longleftrightarrow  C^\infty(\fSpreg0)^\fN.
\label{Sisom2}\ee

We next point out that the bracket \eqref{qSPB}  simplifies considerably for invariant functions.

\begin{prop}
If $\cF, \cH \in C^\infty(\fS)^G$, then the formula \eqref{qSPB} can be rewritten as
\be
2\{ \cF,\cH\}  = \langle \nabla_1 \cH, \nabla_2 \cF  + \nabla_2' \cF\rangle -\langle \nabla_1 \cF, \nabla_2 \cH +\nabla_2' \cH \rangle+
 \langle \nabla_2 \cH, \nabla_2' \cF \rangle  -  \langle \nabla_2' \cH, \nabla_2 \cF  \rangle,
\label{SinvPB}\ee
and alternatively also as
\be
2\{ \cF,\cH\}  =  \langle \nabla_2 \cF, \nabla_1 \cH  + \nabla_1' \cH\rangle -\langle \nabla_2 \cH, \nabla_1 \cF +\nabla_1' \cF \rangle+
 \langle \nabla_1 \cF, \nabla_1' \cH \rangle  -  \langle \nabla_1' \cF, \nabla_1 \cH \rangle.
\label{SinvPBprime}\ee
\end{prop}
\begin{proof}
The derivatives of the $G$-invariant functions $\cF$ and $\cH$ satisfy
\be\label{qderrels}
\nabla_1 \cF -\nabla_1' \cF  + \nabla_2 \cF - \nabla_2' \cF =0
\quad\hbox{and} \quad
\nabla_1\cH - \nabla_1' \cH  + \nabla_2 \cH - \nabla_2' \cH = 0.
\ee
Formula \eqref{SinvPB} results from \eqref{qSPB} by elimination of $\nabla_1'\cF $ and $\nabla_1' \cH$
via these relations, and \eqref{SinvPBprime} results by doing the same to $\nabla_2'\cF $ and $\nabla_2' \cH$.
\end{proof}

Formulae \eqref{SinvPB} and \eqref{SinvPBprime} are also valid for
invariant functions on any open $G$-invariant submanifold of $\fS$. This simple remark is applied below.
Any function $F\in C^\infty(\fS_0^\reg)$ has the $\cG_0$-valued derivative $\nabla_1 F$ and the $\cG$-valued
derivatives $\nabla_2 F$ and $\nabla_2'F$, which are defined in the natural manner.
For functions on $\fSpreg0$, the roles of the subscripts $1$ and $2$ are exchanged.

\begin{thm}
\label{thm:redqPB}
First, let $F,H \in C^\infty(\fS_0^\reg)^\fN$  be  the restrictions of $\cF, \cH \in C^\infty(\fS^\reg)^G$, respectively.
Then, the definition
\be
\{ F,H\}_\red(Q,g) := \{ \cF, \cH\}(Q,g), \qquad \forall (Q,g) \in \fS_0^\reg,
\ee
leads to the formula
\be\label{redqPB}
\{F,H\}_\red(Q,g) = \langle \nabla_1 H, \nabla_2 F  \rangle  -
 \langle \nabla_1 F, \nabla_2 H  \rangle
+ \langle \nabla_2' F, \cR(Q)  \nabla_2' H\rangle  - \langle \nabla_2 F, \cR(Q)  \nabla_2 H  \rangle.
 \ee
Second,
let $F,H \in C^\infty(\fSpreg0)^\fN$  be  the restrictions of $\cF, \cH \in C^\infty(\FSpreg)^G$, respectively.
Then, the definition
\be
\{ F,H\}'_\red(g,Q) := \{ \cF, \cH\}(g,Q), \qquad \forall (g,Q) \in \fSpreg0,
\ee
gives
\be\label{redqPBprime}
\{F,H\}'_\red(g,Q) =    \langle \nabla_2 F, \nabla_1 H \rangle  -
\langle \nabla_2 H, \nabla_1 F  \rangle
 + \langle \nabla_1 F, \cR(Q)  \nabla_1  H\rangle  - \langle \nabla_1' F, \cR(Q)  \nabla_1' H  \rangle.
 \ee
Here, $\cR(Q)$ is given by \eqref{RQ1},  and the derivatives are taken at $(Q,g)$ and at $(g,Q)$,
respectively.
\end{thm}
\begin{proof}
By taking advantage of the identity \eqref{qderrels} at $(Q,g)\in \fS_0^\reg$ \eqref{S01}, we can express the derivatives of $\cF$
in terms of the derivatives of $F$ as follows:
\be
\nabla_2 \cF(Q,g) = \nabla_2 F (Q,g), \qquad \left( \nabla_2  F(Q,g) - \nabla'_2 F (Q,g) \right)_0 =0,
\label{F14}\ee
\be
\nabla_1 \cF(Q,g) = \nabla_1 F(Q,g) - (\cR(Q) + \frac{1}{2} \id)\left( \nabla_2 F(Q,g) - \nabla_2' F(Q,g)   \right).
\label{F15}\ee
By inserting this and the similar formula for the derivatives of $\cH$ into \eqref{SinvPB}, we obtain
\eqref{redqPB}. The details of this straightforward calculation are omitted.
The derivation of \eqref{redqPBprime} is fully analogous, and can also be obtained from the previous case
by  exchange of the subscripts $1$ and $2$, accompanied by applying an overall minus sign.
\end{proof}

\begin{prop}\label{prop:redqeq}
If $H$ is the restriction of $\cH = \pi_2^* \phi$  with $\phi\in C^\infty(G)^G$, then the reduced Poisson bracket \eqref{redqPB} gives
 \be\label{redqPBinv}
\{F,H\}_\red(Q,g) =   - \langle \nabla_1 F(Q,g), \nabla \phi(g)  \rangle
+ \langle \nabla_2' F(Q,g) - \nabla_2 F(Q,g), \cR(Q)  \nabla\phi (g)\rangle,
 \ee
 and if $H$ is the restriction of $\cH = \pi_1^* \phi$  with $\phi\in C^\infty(G)^G$, then the reduced Poisson bracket \eqref{redqPBprime} gives
 \be\label{redqPBprimeinv}
\{F,H\}'_\red(g,Q) =   \langle \nabla_2 F(Q,g), \nabla \phi(g)  \rangle
- \langle \nabla_1' F(Q,g) - \nabla_1 F(Q,g), \cR(Q)  \nabla\phi (g)\rangle.
 \ee
 Therefore, the reduced evolution equations associated with $H$ can be written, respectively, as
 \be
 \dot{Q} = - (\nabla\phi(g))_0 Q, \qquad \dot{g} = [ g, \cR(Q) \nabla\phi(g) ], \quad \hbox{on}\quad \fS_0^\reg,
 \label{qredeq}\ee
 and as
 \be
 \dot{Q} =  (\nabla\phi(g))_0 Q, \qquad \dot{g} = -[ g, \cR(Q) \nabla\phi(g) ], \quad \hbox{on}\quad \fSpreg0.
 \label{qredeqprime}\ee
 \end{prop}

\begin{rem}
It is easily seen that the formulae of Theorem \ref{thm:redqPB} yield Poisson algebra structures on $C^\infty(\fS_0^\reg)^{G_0}$ and,
respectively, on $C^\infty(\fSpreg0)^{G_0}$, too.
These Poisson algebras and also the evolution equations of Proposition \ref{prop:redqeq}
differ only by an overall sign (and the allocation of the labels $1$ and $2$).
They are converted into one another by the map
$(Q,g) \mapsto (g, Q^{-1})$.
Thus the two models of the reduced double that we developed carry equivalent copies of the same system.
\end{rem}

\begin{rem}\label{rem:Cas}
It is known \cite{AKSM} that the Poisson center of the Poisson algebra $C^\infty(\fS)^G$ of the  internally fused double $\fS$ \eqref{D}
is formed by the functions $\cC$ of the form
\be
\cC(g_1,g_2) = h(g_1 g_2 g_1^{-1} g_2^{-1}),
\qquad
h\in C^\infty(G)^G.
\label{Caseq}\ee
By fixing the values of all these Casimir functions,
one obtains a Poisson subspace of $\fS/G$, which is the disjoint union of a dense open symplectic manifold
and lower-dimensional symplectic strata.
The restrictions of the reduced systems on generic symplectic leaves of the reduced double $\fS/G$ are
expected to be integrable in the degenerate sense.  They inherit a large set of integrals of
motion from the unreduced master system, and the same counting arguments work as for the spin Sutherland
models of Section \ref{S:Sec2.1}.
\end{rem}

\begin{rem}
The investigations reported in \cite{FK2,FeKlu} are equivalent to studying particular Poisson subspaces of $\fS/G$ for $G=\mathrm{SU}(n)$.
They can be obtained by fixing the values of the functions $h$ in \eqref{Caseq} so that they define a minimal conjugacy class in $G$,
of dimension $2(n-1)$. The Poisson subspaces in question were shown to be smooth symplectic manifolds, and the reduced
integrable system was interpreted as a compactified trigonometric Ruijsenaars--Schneider model.
\end{rem}

We end by recalling \cite{FK3} that the group $\sl2z$ acts on $\fS/G$ as follows.
Define the diffeomorphisms
$S_\fS$ and $T_\fS$ of the double by
\be
S_\fS(g_1, g_2)= (g_2^{-1}, g_2^{-1} g_1 g_2)\quad\hbox{and}\quad T_\fS(g_1,g_2):= (g_1 g_2, g_2).
\ee
These maps descend to maps $\hat S$ and $\hat T$ of $\fS/G$ that satisfy the identities
\be
{\hat S}^2 = (\hat S \circ \hat T)^3, \qquad {\hat S}^4 = \mathrm{id},
\label{sl2z}\ee
and preserve the Poisson brackets on $C^\infty(\fS/G)\simeq C^\infty(\fS)^G$ as well as the level surfaces of the Casimir functions \eqref{Caseq}.
The identities \eqref{sl2z} represent the standard defining relations of the group $\sl2z$.   In the matrix realization
they are enjoyed by the generators
\be
S=   \begin{bmatrix}
0 &1  \\
-1 & 0
\end{bmatrix}, \qquad
T =  \begin{bmatrix}
1 &0  \\
1 & 1
\end{bmatrix}.
\ee
Notice that $\hat S$ maps into one another the two reduced Abelian Poisson algebras arising from the two sets of pullback invariants.
It can be interpreted as a self-duality map that converts the `global position variables' into the `Hamiltonians of interest'
of the reduced systems that descend from the double.
Referring to Proposition \ref{prop:redqeq},  the `global position variables' are
those functions of $Q$ that are restrictions of pullback invariants,
and the `Hamiltonians of interest' are the $G$-invariant functions of the Lax matrix $g$.

\section{Summary and outlook}
\label{S:Sec5}

In this paper, we performed a systematic study of Poisson reductions of  `master integrable systems' carried by the
classical doubles of any compact (connected and simply connected) Lie group $G$ associated with a simple Lie algebra $\cG$.
Informally, using the terminology of matrix Lie groups, the outcome of
our analysis  can be summarized as follows. The starting phase space always consists of a pair of matrices,
and the action of $G$ is equivalent to  simultaneous conjugation of those two matrices by the elements of $G$.
We proceeded by bringing one of those matrices to a `diagonal' normal form,  and letting the other matrix
serve as a Lax matrix that generates commuting Hamiltonians.
The Lax matrix then satisfies reduced evolution equations of the form
\be
\dot{\cL} = [ \bfR(X) Y(\cL), \cL],
\ee
where $Y(\cL)$ is a $\cG$-valued derivative of a $G$-invariant function of $\cL$, and $\bfR(X)$
is a dynamical $r$-matrix depending on the diagonal `position matrix' $X$.
The aligned evolution equation for $X$ contains the Cartan subalgebra
component $Y(\cL)_0$ of $Y(\cL)$.
The nature of the matrices $\cL$, $X$ and $\bfR(X)$ and
the derivative $Y(\cL)$
varies case by case, and is described in the text.  This description is valid on a dense open
subset, where $X$ satisfies a regularity condition.
The matrices $X$ and $\cL$ are subject to residual gauge transformations by the normalizer of a fixed maximal torus  of $G$,
and we found the explicit formula for the Poisson brackets of the corresponding invariant functions.
The dynamical $r$-matrices play prominent role in the reduced Poisson brackets as well.
In precise technical terms, the reduced evolution equations are given  by equations \eqref{redeq1}, \eqref{redeq2},
\eqref{REDeq1+}, \eqref{REDeq2}, \eqref{qredeq} and \eqref{qredeqprime}.
The reduced Poisson brackets are characterized by  Theorems \ref{thm:spinSuth},  \ref{thm:spinRS}, \ref{thm:defSuth}, \ref{thm:RED2} and
 \ref{thm:redqPB}.

We explained that the unreduced master systems  possess the characteristic properties of degenerate integrability.
Then, we presented convincing arguments indicating  that these properties are inherited by the reduced systems, on generic symplectic
leaves of the reduced Poisson space.
A fully rigorous proof of integrability after reduction
 is hindered by the fact that the orbit space of the $G$-action is  not a smooth manifold.
We conjecture that reduced integrability holds on all symplectic leaves of the quotient space,
generically degenerate integrability, and only Liouville integrability on exceptional symplectic leaves.

On special symplectic leaves of the reduced Poisson spaces associated with $G=\mathrm{SU}(n)$, one recovers
the trigonometric Sutherland and Ruijsenaars--Schneider models,  which are known to be (only) Liouville integrable \cite{RBanff}.
These special cases and the changes of variables discussed around equations \eqref{Jpar}, \eqref{spinSuthHam} and \eqref{defLax} motivated
us to call the reduced systems  spin Sutherland and spin Ruijsenaars--Schneider type models.  This terminology was also used in
the  papers by Reshetikhin \cite{Res1,Res2} dealing with related complex holomorphic systems.

A very interesting open problem concerns the generalization of our analysis to doubles of
loop groups. The investigation of quantum Hamiltonian
reductions corresponding to our classical reductions appears to
be a worthwhile project for the future, too. As far as we know, such a reduction treatment is so far available
(see, e.g., \cite{FP2}) only for the spin Sutherland models descending from $T^*G$.

\bigskip
\subsection*{Acknowledgement}
I wish to thank Maxime Fairon for  useful remarks on the manuscript.
This work was supported in part by the NKFIH research grant K134946.

\bigskip

\appendix
\section{Some Lie theoretic facts}
\label{sec:A}

We here collect a few Lie theoretic  definitions and results, which are used in the main text.
For references, see e.g.~\cite{DK,Knapp,Sam}.

Consider  a compact simple Lie algebra  $\cG$, i.e., a simple real Lie algebra whose Killing form is
negative definite.
Denote $\cG^\bC$ the complex simple Lie algebra
obtained as the complexification of $\cG$.
(Equivalently, one may start with a complex simple Lie algebra and then pick its compact real form.)
Let $\cG^\bC$ carry the normalized Killing form $\langle - ,- \rangle$, given by
\be
\langle Z_1,  Z_2 \rangle = c\, \tr (\ad_{Z_1} \circ \ad_{Z_2}), \quad Z_1, Z_2\in \cG^\bC,
\label{C1}\ee
where $c$ is some convenient, positive constant.
The restriction of $\langle - ,- \rangle$ to $\cG$ is the (normalized)
Killing form $\langle -,- \rangle_\cG$ of $\cG$.
We may regard $\cG^\bC$ also as a real Lie algebra,  in which case we denote it $\cG^\bC_\bR$.
Up to an overall, positive constant,
the Killing form of $\cG^\bC_\bR$ is given by the real part $\langle -,- \rangle_\bR$ of $\langle - , - \rangle$.
The real vector space $\cG^\bC_\bR$ can be written as the direct sum
\be
\cG^\bC_\bR = \cG + \ri \cG,
\ee
since every element $Z \in \cG^\bC_\bR$ can be decomposed uniquely as
\be
Z= X + \ri Y,
\quad X,Y\in \cG.
\ee
By definition,  the complex conjugation on $\cG^\bC_\bR$ with respect to $\cG$ is the map $\theta$ defined by
\be
\theta( X+ \ri Y) := X - \ri Y.
\ee
The complex conjugation $\theta$ is an involutive automorphism of the real Lie algebra $\cG_\bR^\bC$.
It is a \emph{Cartan involution}, since $\langle - , - \rangle_\bR$ is negative definite
on its fixed point set, $\cG$, and is positive definite on  its eigensubspace with eigenvalue $-1$,
$\ri \cG$.
When regarded as a map of $\cG^\bC$ to itself, $\theta$ is conjugate linear, i.e.,
$\theta(\lambda Z) = \bar \lambda \theta(Z)$ for all $\lambda \in \bC$.
Notice also from the definitions that
\be
\langle \theta(Z_1), \theta(Z_2) \rangle = \overline{\langle Z_1, Z_2 \rangle},
\qquad \forall Z_1, Z_2 \in \cG^\bC.
\label{thetaKill}\ee

We also need  the real bilinear form on $\cG_\bR^\bC$ provided by the imaginary part
of the complex Killing form,
\be
\langle Z_1, Z_2 \rangle_\bI := \Im \langle Z_1, Z_2 \rangle.
\ee
As a  result of \eqref{thetaKill},
this invariant, non-degenerate, symmetric bilinear form enjoys the equality
\be
\langle \theta(Z_1), \theta(Z_2) \rangle_\bI = - \langle Z_1, Z_2 \rangle_\bI, \quad \forall Z_1, Z_2\in \cG_\bR^\bC.
\ee

A crucial fact is that $\cG_\bR^\bC$ can be presented as the vector
space direct sum of two isotropic subalgebras with respect to the bilinear form $\langle - , - \rangle_\bI$:
\be
\cG^\bC_\bR = \cG + \cB,
\label{cGcB1}\ee
where $\cB$ is a suitable `Borel' subalgebra.
We next recall how these subalgebras can be described  using the root
space decomposition of $\cG^\bC$.
For this, let us pick a maximal Abelian subalgebra $\cG_0$ of $\cG$. Its complexification $\cG_0^\bC$
is a Cartan subalgebra of $\cG^\bC$.
Then consider the corresponding set of roots, $\fR$, and decompose $\fR$ into sets of
positive and negative roots $\fR^\pm$.
Moreover, let $\Delta$ be the associated set of simple roots.
It is easily seen that the Cartan involution $\theta$ maps any root subspace $\cG^\bC_\alpha$ ($\alpha \in \fR$) to
$\cG^\bC_{-\alpha}$.

We then choose a Weyl--Chevalley basis of $\cG^\bC$, which  consists of root vectors
$E_\alpha$ for which $\langle E_\alpha, E_{-\alpha} \rangle = 2/\vert \alpha \vert^2$ for all $\alpha\in \fR^+$,
and Cartan elements $H_{\alpha_j}:= [E_\alpha, E_{-\alpha_j}]$ for $\alpha_j \in \Delta$.
The root vectors are chosen in such a way that
all structure constants are real   and $E_{-\alpha} = - \theta(E_\alpha)$ holds.
(Then, if $\alpha, \beta$ and $(\alpha + \beta)$ are roots, one has $[E_\alpha, E_\beta ] = N_{\alpha, \beta} E_{\alpha+ \beta}$
and $[E_{-\alpha}, E_{-\beta} ] = -N_{\alpha, \beta} E_{-\alpha - \beta}$; and all structure constants are integers \cite{Sam}.)
Using any such basis, $\cG$ is given by
\be
\cG = \span_{\bR}\{\ri H_{\alpha_j},\, (E_\alpha -E_{-\alpha}),\, \ri(E_\alpha + E_{-\alpha}) \mid \alpha_j\in \Delta,\, \alpha\in \fR^+\},
\ee
and one can take
\be
\cB = \span_{\bR}\{ H_{\alpha_j},\, E_\alpha,\, \ri E_\alpha  \mid \alpha_j\in \Delta,\, \alpha\in \fR^+\}.
\ee
It is worth noting that there are as many choices for $\cB$ as systems of positive roots,
but all of them are equivalent by the action of the Weyl group of the root system.

Next,  we explain why the map $\nu: B\to \fP$ \eqref{nu} is a diffeomorphism. To start, define the
maps
\be
\mu_1: G_\bR^\bC/G \to B
\quad\hbox{and}\quad \mu_2: G_\bR^\bC/G \to \fP
\ee
by
\be
\mu_1: [K] \mapsto \Lambda_L(K)\quad \hbox{and}\quad
\mu_2: [K] \mapsto K K^\tau,
\ee
where $[K] = K G \in G_\bR^\bC/G$, $\forall K\in G^\bC_\bR$, and we used the definitions \eqref{XiLaT} and \eqref{taumap}.
Recall that $G^\bC_\bR$ is diffeomorphic to $B\times G$ and to $\fP \times G$ by the Iwasawa and global Cartan decompositions, respectively,
and $\fP=\exp(\ri \cG)$ is diffeomorphic to $\ri\cG$ by the exponential map \cite{Knapp}. It follows that $\mu_1$ and $\mu_2$
are (real analytic) diffeomorphisms with the inverses
\be
\mu_1^{-1}: b \mapsto bG,\,\, \forall b\in B \quad\hbox{and}\quad  \mu_2^{-1}: P \mapsto \sqrt{P} G,
\,\,
\forall P\in \fP.
\ee
Therefore the composed map $\nu = \mu_2 \circ \mu_1^{-1}: B \to \fP$ is a diffeomorphism, with the inverse
operating as
$\nu^{-1}: P \mapsto \Lambda_L( \sqrt{P})$.

At the end, we present some remarks on the rings of $G$-invariant functions on which
our integrable systems are  based.  Here, the following isomorphisms are fundamental:
\be
C^\infty(\cG)^G \longleftrightarrow  C^\infty(\cG_0)^\cW
\quad\hbox{and}\quad
 C^\infty(G)^G \longleftrightarrow C^\infty(G_0)^\cW,
\ee
where $\cW$ is the Weyl group.  These are generalizations \cite{Mic,PT} of the Chevalley isomorphism theorem between $G$-invariant
polynomials on $\cG$ and $\cW$-invariant polynomials on the Cartan subalgebra $\cG_0$.  The isomorphisms result from the pertinent
restrictions of functions, and they readily imply that both $C^\infty(\cG)^G$ and   $C^\infty(G)^G$ have functional dimension
$\ell= \rank(G)$.
By combining a theorem of \cite{Schw}  on smooth invariants with the fact that the ring of $G$-invariant polynomials on $\cG$ is freely generated by
$\ell$ homogeneous polynomials, $\sigma_1,\dots, \sigma_\ell$, one obtains that $C^\infty(\cG)^G$ consists of the functions $\phi$ of the form
$\phi = f(\sigma_1,\dots, \sigma_\ell)$ with arbitrary $f\in C^\infty(\bR^\ell)$.
This  gives the structure of the ring $C^\infty(B)^G$, too, by utilizing the isomorphisms
\be
C^\infty(B)^G \longleftrightarrow C^\infty(\fP)^G \longleftrightarrow C^\infty(\cG)^G,
\ee
which arise from the $G$-equivariant diffeomorphism $\nu$ \eqref{nu} and the  exponential parametrization of $\fP = \exp(\ri \cG)$.

Let $\rho: G \to \mathrm{GL}(V)$ be an irreducible unitary representations of $G$, and $\varrho$ the corresponding representation of
$\cG$. Then, the character
$G\ni g\mapsto \tr \rho(g)$ is a $G$-invariant (in general complex) function on $G$, and $\fP\ni e^{\ri X} \mapsto \tr e^{\ri  \varrho(X)}$ is
a $G$-invariant real function on $\fP$.
By taking suitable real or imaginary parts, it should be possible to obtain $\ell$ functionally independent elements of $C^\infty(G)^G$
from the fundamental irreducible representations of $G$.  In the case of $\cG$,
 the real trace functions $\cG\ni X\ \mapsto \tr(\ri \varrho(X))^k$, with $k\geq 2$, provide convenient invariants.

\section{Equivalence of two models of the Heisenberg double}
\label{sec:B}

According to the original definition \cite{STS}, the Heisenberg double of the Poisson--Lie group $G$
is the Poisson  manifold $(M, \{-,- \}_+)$, where $M= G^\bC_\bR$ and
for any $F, H \in C^\infty(M)$
\be
\{ F, H\}_+ = \langle \nabla F, \rho \nabla H \rangle_\bI +  \langle \nabla' F, \rho \nabla' H \rangle_\bI,
\label{A1}\ee
with  $\rho := \frac{1}{2}\left( \pi_{\cG} - \pi_{\cB}\right)$
defined with the aid of the vector space direct sum $\cG^\bC_\bR = \cG + \cB$.
The corresponding symplectic form was found in \cite{AM}.
An alternative model of this Poisson space is $(\fM,\{- ,- \})$, where  $\fM = G \times B$ and
\be
\{\cF, \cH\}(g,b) =\left\langle D_2' \cF, b^{-1} (D_2\cH) b \right\rangle_\bI
-\left\langle D'_1\cF, g^{-1} (D_1\cH) g\right\rangle_\bI
 +  \left\langle D_1\cF , D_2\cH \right\rangle_\bI
-\left\langle D_1 \cH , D_2\cF \right\rangle_\bI
\label{A2}\ee
for functions $\cF, \cH$ on $\fM$.
The derivatives on the right-hand side
are   taken at $(g,b)\in G\times B$, with respect to the first and second variable, respectively.
See also equations  \eqref{Nab}, \eqref{derB} and \eqref{derG} for the definitions of the derivatives.

The purpose of this appendix
 is to explain that the bracket \eqref{A2} on $\fM$ is the pushforward of the standard Poisson bracket \eqref{A1} by
the diffeomorphism $m$ \eqref{mMfM} between $M$ and $\fM$. In particular, this proves that $(\fM, \{-,- \})$ is indeed
a Poisson manifold.

\begin{lem} Using the definitions \eqref{XiLaT}, the map $m: M \to \fM$ given by
\be
m= (\Xi_R, \Lambda_R) \quad \hbox{that is}\quad m(K) = (g_R, b_R)
\label{A4}\ee
is a real analytic diffeomorphism.
\end{lem}
\begin{proof}
For any $K\in M$, the unique Iwasawa decompositions $K= b_L g_R^{-1} = g_L b_R^{-1}$
\eqref{KdecT} imply the equality $g_L^{-1} b_L = b_R^{-1} g_R$.  This shows that $b_L\in B$ and $g_L\in G$,
and thus also $K$, can be
recovered from $b_R \in B$ and $g_R \in G$. Hence, the map $m$ is injective.
The surjectivity of the map $m$ is also clear, since by re-decomposing $b_R^{-1} g_R$ in the other order we can construct
$K$ such that $(g_R, b_R) = m(K)$.  The real analytic nature of the relevant decompositions is well known \cite{Knapp}.
\end{proof}

Let $\pi_1$ and $\pi_2$ denote the obvious projections from $\fM= G\times B$ onto $G$ and $B$, respectively.
Then we have the identities
\be
\Xi_R = \pi_1\circ m,
\qquad
\Lambda_R = \pi_2 \circ m.
\label{A5}\ee
We wish to prove that
\be
\{ \cF, \cH\} \circ m = \{ \cF \circ m, \cH\circ m\}_+,
\qquad
\forall \cF, \cH \in C^\infty(\fM).
\label{A6}\ee
We start with  two useful lemmas.

\begin{lem}\label{lm:B2} For any $f\in C^\infty(G)$ and $\varphi \in C^\infty(B)$, consider
the functions $f \circ \Xi_R$ and $\varphi \circ \Lambda_R$ on $M$.
Then, the derivatives of these functions obey the identities
\be
(\nabla' \varphi \circ \Lambda_R)(K) = - b_R D' \varphi(b_R) b_R^{-1},
\qquad
(\nabla' f \circ \Xi_R)(K) = - g_R D' f(g_R) g_R^{-1},
\label{A7}\ee
and
\be
(\nabla \varphi \circ \Lambda_R)(K) = - g_L  (D' \varphi(b_R)) g_L^{-1},\qquad
(\nabla f \circ \Xi_R)(K)  = - b_L  (D' f(g_R)) b_L^{-1},
\label{A8}\ee
where the decompositions $K = b_L g_R^{-1} = b_L g_R^{-1}$ \eqref{KdecT} are used.
\end{lem}
\begin{proof}
Denote $F:= \varphi \circ \Lambda_R$ and use the decompositions of $K\in M$ defined in \eqref{KdecT}.
Then, for any $X\in \cB$ and $K\in M$, we have
\be
\langle \nabla' F(K), X\rangle_\bI = \dt F(K e^{tX}) = \dt F(g_L b_R^{-1} e^{tX} ) = \dt \varphi (e^{- tX} b_R )  = - \langle X, D\varphi(b_R)\rangle_\bI,
\ee
which means that
\be
(\nabla' F(K))_\cG = - D\varphi(b_R) = - \left(b_R (D' \varphi (b_R)) b_R^{-1}\right)_\cG,
\ee
where the second equality reflects the relation of the left- and right-derivatives of $\varphi$.
Next, taking $X\in \cG$, notice from the definitions \eqref{XiLaT} and \eqref{Dress} that
\be
\Lambda_R( K e^{tX} ) = \Lambda_R ( g_L (e^{-tX} b_R)^{-1}) = \Dress_{e^{-tX}}(b_R),
\ee
and therefore
\be
\langle \nabla' F(K), X\rangle_\bI  = \dt \varphi( \Dress_{e^{-tX}}(b_R)) = \langle D'\varphi(b_R), - (b_R^{-1} X b_R)_\cB \rangle_\bI =
- \langle b_R D' \varphi (b_R) b_R^{-1}, X  \rangle_\bI,
\label{seceq}\ee
which means that
\be
(\nabla' F(K))_\cB = -\left(b_R (D' \varphi (b_R)) b_R^{-1}\right)_\cB.
\ee
The second equality in \eqref{seceq} follows from formula \eqref{dressT} of the infinitesimal dressing action.
Putting these together, we have proved the first relation in \eqref{A7}, and the second one is derived in a similar manner.
These imply the equalities \eqref{A8} since, for any function $F$ on $M$, $\nabla F(K) = K \nabla' F(K) K^{-1}$.
 \end{proof}

Let us recall that $(G, \{-,- \}_G)$ and $(B,\{-,- \}_B)$ are Poisson--Lie groups, with the
Poisson structures
\be
\{ \varphi_1, \varphi_2\}_B(b) = \langle D' \varphi_1(b), b^{-1} (D \varphi_2(b)) b \rangle_\bI
\quad\hbox{and}\quad
\{ f_1, f_2\}_G(g) = - \langle D' f_1(g), g^{-1} (D f_2(g)) g \rangle_\bI.
\label{PBBG}\ee

Based on the above definitions and the relations of the various derivatives,
the following statement is easily verified.

\begin{lem}\label{lm:B3} For arbitrary smooth functions $\varphi_i$ on $B$ and $f_i$ on $G$ ($i=1,2$), we have
\be
\{ \varphi_1 \circ \Lambda_R, \varphi_2 \circ\Lambda_R\}_+ = \{ \varphi_1, \varphi_2\}_B \circ \Lambda_R,
\quad
\{ \varphi_1 \circ \pi_2, \varphi_2 \circ \pi_2\} = \{ \varphi_1, \varphi_2\}_B \circ \pi_2,
\label{LB41}\ee
\be
\{ f_1 \circ \Xi_R, f_2 \circ\Xi_R\}_+ = \{ f_1, f_2\}_G \circ \Xi_R,
\quad
\{ f_1 \circ \pi_1, f_2 \circ \pi_1\} = \{ f_1, f_2\}_G \circ \pi_1,
\label{LB42}\ee
and
\be
\{ f_i \circ \Xi_R, \varphi_j \circ\Lambda_R\}_+ =\langle (Df_i)\circ \Xi_R, (D \varphi_j) \circ \Lambda_R \rangle_\bI,
\quad
\{ f_i \circ \pi_1, \varphi_j \circ\pi_2\} =\langle (Df_i)\circ \pi_1, (D \varphi_j) \circ \pi_2 \rangle_\bI.
\label{LB43}\ee
\end{lem}
\begin{proof}
For example, let us consider  arbitrary
$f\in C^\infty(G)$ and $\varphi\in C^\infty(B)$.  Then, due to Lemma \ref{lm:B2}, the first term in the formula \eqref{A1} gives
\be
\begin{aligned}
&\langle \nabla f\circ \Xi_R (K), \rho \nabla \varphi \circ \Lambda_R (K) \rangle_\bI =
\frac{1}{2}  \langle b_L  D' f(g_R) b_L^{-1}, g_L  D' \varphi(b_R) g_L^{-1}\rangle_\bI
 = \frac{1}{2}  \langle g_L^{-1} b_L  D' f(g_R) b_L^{-1} g_L,  D' \varphi(b_R)  \rangle_\bI \\
&= \frac{1}{2}  \langle b_R^{-1} g_R   D' f(g_R) g_R^{-1} b_R,  D' \varphi(b_R)  \rangle_\bI
 =  \frac{1}{2}  \langle  g_R  D' f(g_R) g_R^{-1} ,  b_R D' \varphi(b_R) b_R^{-1}  \rangle_\bI \\
 &= \frac{1}{2}  \langle   D f(g_R)  ,   D \varphi(b_R)   \rangle_\bI +
 \frac{1}{2} \langle  (g_R  D' f(g_R) g_R^{-1})_\cG ,  (b_R D' \varphi(b_R) b_R^{-1})_\cB  \rangle_\bI.
\end{aligned}
\ee
On the other hand, the second term gives
\be
\begin{aligned}
&\langle \nabla' f\circ \Xi_R (K), \rho \nabla' \varphi \circ \Lambda_R (K) \rangle_\bI =
  \langle g_R  D' f(g_R) g_R^{-1}, \rho(b_R  D' \varphi(b_R) b_R^{-1})\rangle_\bI \\
  \quad &= \frac{1}{2}\langle  D f(g_R) ,  D \varphi(b_R) \rangle_\bI -
  \frac{1}{2}  \langle (g_R  D' f(g_R) g_R^{-1})_\cG, (b_R  D' \varphi(b_R) b_R^{-1})_\cB\rangle_\bI.
\end{aligned}
\ee
Combining these terms, we obtain the first identity in \eqref{LB43}. The rest of the identities follows by similar, but shorter, calculations.
\end{proof}

\begin{rem} Lemma \ref{lm:B3} says, in particular, that $\Lambda_R$ and $\Xi_R$ are Poisson maps from $(M, \{- ,- \}_+)$ to
$B$ and $G$ equipped with the
Poisson structures \eqref{PBBG} on $B$ and $G$, respectively. One can show that $\Lambda_L$ and $\Xi_L$
have the same properties. Moreover,
\be
\{ \varphi_1 \circ \Lambda_L, \varphi_2 \circ \Lambda_R\}_+ = \{ f_1 \circ \Xi_L , f_2\circ \Xi_R\}_+ =0
\ee
holds for all $\varphi_i\in C^\infty(B)$ and $f_i\in C^\infty(G)$.
These statements follow also from the general theory of the Heisenberg double \cite{STS,STS+}.
\end{rem}

\begin{prop}
The map $m$ \eqref{A4} is a Poisson diffeomorphism between $(M,\{- ,- \}_+)$ \eqref{A1} and $(\fM, \{- ,- \})$ \eqref{A2},
that is, the equality \eqref{A6} holds.
\end{prop}
\begin{proof}
Notice that the equality \eqref{A6} follows for all smooth functions on $\fM$ if we prove it for those functions
that are of the form
$f \circ \pi_1$ and $\varphi \circ \pi_2$ for arbitrary smooth functions $f$ on $G$ and $\varphi$ on $B$.
In order to see this, it is enough to remark that the exterior derivatives of such functions span
the cotangent space to $\fM$ at any point.

For the types of functions that feature in Lemma \ref{lm:B3}, using also \eqref{A5}, we can write
\be
\{ \varphi_1 \circ \pi_2, \varphi_2 \circ \pi_2\}\circ m = \{ \varphi_1, \varphi_2\}_B \circ \pi_2 \circ m =
\{ \varphi_1, \varphi_2\}_B \circ \Lambda_R = \{ \varphi_1\circ \Lambda_R, \varphi_2\circ \Lambda_R \}_+ =
 \{ \varphi_1 \circ \pi_2\circ m, \varphi_2 \circ \pi_2\circ m\}_+.
\ee
The other cases of functions are handled in exactly the same way.
\end{proof}

\section{On the construction of $G$-invariant constants of motion via averaging}
\label{sec:C}

Let $X$ be a $G$-manifold and $V$ a $G$-invariant vector field on $X$,
\be
V = (A_\eta)_* V, \qquad \forall \eta \in G,
\ee
where $A_\eta$ denotes the diffeomorphism of $X$ associated with $\eta\in G$.
The $G$-invariance of the vector field is equivalent  to the property that if
$x(t)$ is an integral curve of $V$,  then $A_\eta(x(t))$
is also an integral curve, for each $\eta \in G$.
Suppose now that $G$ is compact and denote by $d_G$ the Haar measure normalized so
that the volume of $G$ is 1.
For any real function $\cF \in C^\infty(X)$ define the function $\cF^G$ by averaging
the functions $A_\eta^* \cF$ over $G$,
\be
\cF^G(x):= \int_G \cF(A_\eta(x)) d_G(\eta), \qquad \forall x\in X.
\ee
It is clear that $\cF^G\in C^\infty(X)^G$.
Moreover, if $\cF$ is a constant of motion for the vector field $V$, then
$\cF^G$ is also a constant of motion for $V$. Indeed,
for any integral curve $x(t)$
\be
\frac{d}{dt} \cF^G(x(t)) = \int_G \frac{d}{dt} \cF(A_\eta(x(t)) d_G(\eta) =0,
\ee
since $A_\eta(x(t))$ is an integral curve for all $\eta$.
In \cite{J,Zung}  this mechanism was used for arguing that, generically,
degenerate integrability  survives Hamiltonian reduction.
In these papers the starting point was a Hamiltonian
action on a symplectic manifold, in which case the Hamiltonian vector fields
of the $G$-invariant Hamiltonians are $G$-invariant.

The averaging of the constants of motion is applicable to the unreduced
integrable systems of our interest if the relevant unreduced
evolution vector fields are $G$-invariant.
This obviously holds for the two degenerate integrable systems on $T^*G$
considered in Section \ref{S:Sec2}, and is also easily checked for the unreduced
evolution vector fields on the quasi-Poisson double $\fS$ studied in Section \ref{S:Sec4}.
The Hamiltonian vector fields  belonging to our master systems on the Heisenberg double $\fM$
enjoy the relevant invariance property \emph{with respect to the quasi-adjoint action} $\cA^2$ \eqref{cA2}.
Indeed,  Hamiltonians invariant with respect to a Poisson action of $G$ always
generate invariant Hamiltonian vector fields. This follows, for example, from  Proposition 5.12 in \cite{STS+}.

For completeness, we below answer the question whether the invariance  property
holds for the Hamiltonian vector fields associated
with the two sets of pullback invariants on $\fM$,  \emph{with respect to
`conjugation action' $\cA$ of $G$ on $\fM$ defined by equation \eqref{cA}}.

\begin{prop}\label{prop:Dphivar}
The derivatives of any $\phi \in C^\infty(B)^G$ satisfy the relations
\be
D'\phi(\Dress_\eta(b)) = \Xi_R(\eta b)^{-1} D'\phi(b) \Xi_R(\eta b),
\quad D\phi(\Dress_\eta(b)) = \eta  D\phi(b) \eta^{-1},
\quad \forall \eta\in G,\,b\in B.
\label{Dphivar}\ee
As a consequence, if $(g(t), b(t))$ is an integral curve of the Hamiltonian
vector field of
$\cH=\pi_2^*\phi \in C^\infty(\fM)^G$, then $\cA_\eta(g(t), b(t))$ is also
an integral curve (with the $G$-action defined in \eqref{cA}).
\end{prop}
\begin{proof}
The result will follow by taking the $t$-derivative of the identity
\be
\phi(b e^{tX}) = \phi(\Dress_\eta(b e^{tX})),\qquad \forall b\in B,\,X\in \cB,\, t\in \bR.
\ee
We see directly from the definitions that
\be
\Dress_\eta(b e^{tX}) = \Dress_\eta(b) \Dress_{\Xi_R(\eta b)^{-1}}(e^{tX}),
\ee
and therefore we get
\be
\langle D'\phi(b), X \rangle_\bI= \langle D'\phi(\Dress_\eta(b)), \dt  \Dress_{\Xi_R(\eta b)^{-1}}(e^{tX})\rangle_\bI.
\ee
Now, we have
\be
\Xi_R(\eta b)^{-1} e^{tX} =  \Dress_{\Xi_R(\eta b)^{-1}}(e^{tX}) \Xi_R (\Xi_R(\eta b)^{-1} e^{tX})^{-1},
\ee
and, at $t=0$,
\be
\Xi_R (\Xi_R(\eta b)^{-1} )^{-1}  =\Xi_R(\eta b)^{-1}, \qquad  \Dress_{\Xi_R(\eta b)^{-1}}(\1_B) = \1_B.
\ee
Hence, taking the derivative at $t=0$ gives
\be
 \dt  \Dress_{\Xi_R(\eta b)^{-1}}(e^{tX}) = \left(\Xi_R(\eta b)^{-1} X \Xi_R(\eta b)\right)_\cB .
\ee
So far we obtained
\be
\langle D'\phi(b), X \rangle_\bI= \langle D'\phi(\Dress_\eta(b)), \left(\Xi_R(\eta b)^{-1} X \Xi_R(\eta b)\right)_\cB \rangle_\bI
= \langle \Xi_R(\eta b)D'\phi(\Dress_\eta(b)) \Xi_R(\eta b)^{-1}, X\rangle_\bI,
\ee
which is equivalent to the equivariance property of $D'\phi$  \eqref{Dphivar}.
Regarding $D\phi$,
we have seen in equation \eqref{S6}  that for the $G$-invariant functions on $B$
\be
D\phi(b) = b D'\phi(b) b^{-1}.
\ee
By combining this with the transformation property of $D'\phi$, we get
\be
D\phi(\Dress_\eta(b)) = \left(\Dress_\eta(b) \Xi_R(\eta b)^{-1} b^{-1}\right) D\phi(b)  \left(\Dress_\eta(b) \Xi_R(\eta b)^{-1} b^{-1}\right)
= \eta D\phi(b) \eta^{-1},
\ee
simply since $\Dress_\eta(b) \Xi_R(\eta b)^{-1} b^{-1} = \eta$.

Next, recall from Proposition \ref{prop:sol1} (equation \eqref{S4}) that the integral curves of $\cH= \pi_2^* \phi$ read
\be
(g(t), b(t)) = \left( \exp\left(t D\phi(b(0))\right) g(0), b(0) \right).
\ee
Therefore
\be
\begin{aligned}
\cA_\eta(g(t), b(t)) &= \left( \eta \exp\left(t D\phi(b(0))\right) g(0) \eta^{-1}, \Dress_\eta(b(0)) \right)\\
&= \left( \exp\left(t \eta D\phi(b(0)) \eta^{-1}\right) \eta g(0) \eta^{-1}, \Dress_\eta(b(0)) \right)\\
&=  \left( \exp\left(t D\phi(\Dress_\eta(b(0)))\right) \eta g(0) \eta^{-1}, \Dress_\eta(b(0)) \right),
\end{aligned}
\ee
which is the integral curve through the initial value $\cA_\eta (g(0), b(0))$.
\end{proof}

 We see from Proposition \ref{prop:Dphivar} that taking the $G$-average of an arbitrary constant of motion
using the conjugation action $\cA$ \eqref{cA}
yields a $G$-invariant constant of motion for the degenerate integrable system
on the Heisenberg double whose Hamiltonians arise from $C^\infty(B)^G$.
However, as we shall see below, the Hamiltonian vector fields stemming from $C^\infty(G)^G$ do not
have the relevant invariance property \emph{with respect to this action}.

For any $h\in C^\infty(G)^G$, the integral curves $(g(t), b(t))\in \fM$ of the Hamiltonian
$\cH = \pi_1^* h\in C^\infty(\fM)^G$  can be read off from equation \eqref{s3} in  Proposition \ref{prop:sol2}.
The identification $(g(t), b(t)) = (g_R(t), b_R(t))$ gives
\be
(g(t), b(t)) = \left( \gamma(t) g(0) \gamma(t)^{-1}, \beta(t)^{-1} b(0)\right),
\label{F1}\ee
where $(\gamma(t), \beta(t)) \in G \times B$ is defined by
\be
\exp(\ri t \nabla h(g(0))) = \beta(t) \gamma(t).
\label{F2}\ee
We are going to prove the following result.

\begin{prop}\label{prop:eqvarprop}
Let $(g(t), b(t))$ be the integral curve \eqref{F1}
of the Hamiltonian vector field of the pullback invariant $\cH = \pi_1^* h\in C^\infty(\fM)^G$
  associated with the initial value $(g(0), b(0))$.
Then the integral curve associated with the transformed initial value
\be
\left(\eta g(0) \eta^{-1}, \Dress_\eta(b(0))\right), \qquad \eta \in G,
\label{F3}\ee
is given by
\be
\cA_{\Xi_R(\eta \beta(t))^{-1}} \left( g(t), b(t)\right),
\label{F4}\ee
where $\beta(t)$ is the determined by the initial value $g(0)$ via the factorization \eqref{F2}.
\end{prop}

\begin{proof}
Denote by $\tilde \beta, \tilde \gamma$ the solution of the factorization \eqref{F2} at the transformed initial value:
\be
\exp(\ri t \nabla h(\eta g(0) \eta^{-1}))= \tilde \beta(t) \tilde \gamma(t).
\label{F5}\ee
Since $\nabla h$ is $G$-equivariant, we get
\be
\tilde \beta(t) = \Dress_\eta(\beta(t))
\quad \hbox{and}\quad
\tilde \gamma(t) = \Xi_R( \eta \beta(t))^{-1} \gamma(t) \eta^{-1}.
\label{F6}\ee
Therefore, the integral curve $(\tilde g(t), \tilde b(t))$ associated with the transformed initial value
can be written as
\be
\tilde g(t) = \tilde \gamma(t) \eta g(0) \eta^{-1} \tilde \gamma(t)^{-1}
= \Xi_R( \eta \beta(t))^{-1} g(t) \Xi_R( \eta \beta(t)).
\label{F7}\ee
This proves half of our claim. To prove the other half, we inspect
\be
\tilde b(t)= (\Dress_\eta(\beta(t)))^{-1} \Dress_\eta(b(0)).
\label{F8}\ee
Now,
\be
(\Dress_\eta(\beta(t)))^{-1} = \Xi_R( \eta \beta(t))^{-1} \beta(t)^{-1} \eta^{-1}.
\label{F9}\ee
Consequently,
\be
\begin{aligned}
\tilde b(t) & = \Xi_R( \eta \beta(t))^{-1} \beta(t)^{-1} \eta^{-1} \Dress_\eta(b(0)) \\
&= \Xi_R( \eta \beta(t))^{-1}  b(t) b(0)^{-1} \eta^{-1} \Dress_\eta(b(0))\\
&= \left(\Dress_{\Xi_R( \eta \beta(t))^{-1}}(b(t)) \right)   \Xi_R( \Xi_R( \eta \beta(t))^{-1} b(t))^{-1}
b(0)^{-1} \eta^{-1} \Dress_\eta(b(0)).
\end{aligned}
\label{F10}\ee
Furthermore,
\be
\begin{aligned}
\Xi_R( \eta \beta(t))^{-1} b(t) &= (\Dress_\eta \beta(t))^{-1} \eta \beta(t) b(t)\\
&= (\Dress_\eta \beta(t))^{-1} \eta b(0) \\
&= (\Dress_\eta \beta(t))^{-1} (\Dress_\eta b(0)) \Xi_R(\eta b(0))^{-1},
\end{aligned}
\label{F11}\ee
and hence
\be
\Xi_R (\Xi_R( \eta \beta(t))^{-1} b(t))^{-1} = \Xi_R(\eta b(0))^{-1}.
\label{F12}\ee
Plugging this back into the last line of \eqref{F10}  gives
\be
\tilde b(t)  = \left(\Dress_{\Xi_R( \eta \beta(t))^{-1}}(b(t)) \right)
\Xi_R(\eta b(0))^{-1}
 (\eta b(0))^{-1} \Dress_\eta(b(0)) = \Dress_{\Xi_R( \eta \beta(t))^{-1}}(b(t)),
\label{F13}\ee
which finishes the proof.
\end{proof}

\begin{rem}
Proposition \ref{prop:eqvarprop} shows that $\cA_\eta$ \eqref{cA} does not map the pertinent
integral curves \eqref{F1} onto integral curves.
At the same time, it confirms
that changing the initial value by the $G$-action \eqref{cA} does not effect
the projection of the integral curve to the quotient space $\fM/G$.
This is equivalent to the fact that the Hamiltonian vector field $V$ of $\cH = \pi_1^* h$, for $h\in C^\infty(G)^G$,
satisfies
\be
(\cA_\eta)_* V = V + Z,
\ee
where the vector field $Z$ is tangent to the $G$-orbits.
One could find $Z$ explicitly, if desired.
It is worth observing from this state of affairs  that the  action \eqref{cA} on $\fM$  is not a Poisson action.
\end{rem}

\end{document}